\documentclass[twoside]{article}

\usepackage[accepted]{aistats2022}
\usepackage{relsize}

\usepackage[square,authoryear]{natbib}
\usepackage[colorlinks=true,linkcolor=red, citecolor=blue]{hyperref}
\usepackage[utf8]{inputenc} 
\usepackage[T1]{fontenc}    
\usepackage{url}            
\usepackage{booktabs}       
\usepackage{amsfonts}       
\usepackage{nicefrac}       
\usepackage{microtype}      
\usepackage{xcolor}         
\usepackage{amsthm}
\usepackage{amssymb}
\usepackage{algorithm}
\usepackage{algorithmic}
\usepackage{bbold}
\usepackage{dsfont}

\usepackage{graphicx}
\usepackage{subfigure}
\usepackage{centernot}
\usepackage{mathtools}
\usepackage{bbm}

\usepackage{wrapfig}
\usepackage{enumerate}

\newtheorem{Theorem}{Theorem}

\newtheorem{Corollary}{Corollary}
\newtheorem{Definition}{Definition}
\newtheorem{Proposition}{Proposition}

\newcommand{\R}{\mathbb{R}} 
\renewcommand{\H}{\mathcal{H}} 
\newcommand{\X}{\mathcal{X}} 
\newcommand{\Z}{\mathcal{Z}} 
\newcommand{\A}{\mathcal{A}} 
\newcommand{\KSD}{\operatorname{KSD}} 
\newcommand{\GKSD}{\operatorname{Sf-KSD}} 

\newcommand{\E}{\mathbb{E}} 
\newcommand{\T}{\mathcal{T}}
\newcommand{\M}{\mathcal{M}}
\renewcommand{\L}{\mathcal{L}}
\newcommand{\f}{\mathbf{f}}
\newcommand{\g}{\mathbf{g}}

\newcommand{\Var}{\mathbb{V}ar}
\newcommand{\ind}{\mathds{1}}

\usepackage{todonotes}

\begin{document}

%
\runningtitle{Sf-KSD: A Unifying View on KSD Goodness-of-fit Testing }

%

\twocolumn[

\aistatstitle{
\vskip 0.05in
Standardisation-function Kernel Stein Discrepancy:
A Unifying View
on Kernel Stein Discrepancy Tests
for 
Goodness-of-fit 
}
\aistatsauthor{Wenkai Xu}
\aistatsaddress{Department of Statistics,\\
Oxford University}]

\begin{abstract}
Non-parametric goodness-of-fit testing procedures based on kernel Stein discrepancies (KSD) are promising approaches to validate general unnormalised distributions in various scenarios. 
Existing works focused on studying 
kernel choices to boost test performances.
However, the choices of (non-unique) Stein operators 
also have considerable effect on the test performances.
Inspired by the standardisation technique
that was originally developed 
to better derive approximation properties for normal distributions, we  
present a unifying framework, called standardisation-function kernel Stein discrepancy (Sf-KSD), to 
study different Stein operators in 
KSD-based tests for goodness-of-fit. 
We derive explicitly how the proposed 
framework relates to existing KSD-based 
tests 
and show that Sf-KSD 
can be used as a guide to develop novel kernel-based non-parametric tests
on complex data scenarios, e.g. truncated distributions or compositional data. Experimental results demonstrate that the proposed tests control type-I error well and achieve higher test power than existing approaches.
\end{abstract}

\section{INTRODUCTION}
Stein's method \citep{barbour2005introduction} provides an elegant probabilistic tool for characterising and comparing distributions.
Relevant techniques have been used to tackle various problems in statistical inference, random graph theory, and computational biology. 
Modern machine learning tasks may extensively involve modelling and learning with intractable densities, where the normalisation constant (or partition function) is unable to be obtained in closed form, e.g. density estimation \citep{hyvarinen2005estimation}, model criticism \citep{kim2016examples,Lloyd2015}, or generative modelling \citep{goodfellow2014generative}.
\emph{Stein operators} may only require to access distributions through the differential (or difference) of log density functions\footnote{Derivative of log density is also known as score-function.}(or mass functions), which avoids the knowledge of normalisation constant. 
These Stein operators are particularly useful 
to study unnormalised models and recently caught attentions from the machine learning community 
in many aspects \citep{anastasiou2021stein} such as density estimation \citep{barp2019minimum}, 
sampling techniques \citep{chen2018stein, chen2019stein, gorham2019measuring,  oates2019convergence, shi2021sampling},
numerical methods
\citep{barp2018riemannian}, 
approximate inference \citep{huggins2018random, liu2016stein},
and Bayesian inference \citep{fisher2020measure, liu2018riemannian}.

The goodness-of-fit testing procedure aims to check the null hypothesis $\mathrm{H_0}: q=p$, where $q$ is the \emph{known} target distribution and $p$ is the \emph{unknown} data distribution accessible only through a set of samples\footnote{As only one set of samples are observed, the goodness-of-fit testing sometimes is also referred to as the \emph{one-sample} test or {one-sample problem}. This is opposed to the two-sample problem where the distribution $q$ is also unknown and appears in the sample form. 
}, $x_1,\dots,x_n \sim p$.
The non-parametric 
testing refers to the scenario where the assumptions made on the distributions $p$ and $q$ are minimal, i.e. the distributions in non-parametric testing are not assumed to be in any parametric family. 
By contrast, parametric tests (e.g. student t-test or normality test) assume pre-defined parametric family to be tested against and usually deal with summary statistics such as means or standard deviations, which can be restrictive in terms of comparing full distributions. Kernel-based methods have been applied to compare distributions via rich-enough reproducing kernel Hilbert spaces (RKHS) \citep{RKHSbook} and achieved state-of-the-art results for 
non-parametric 
two-sample test \citep{gretton2012kernel} or independence test \citep{gretton2008kernel}.
With well-defined \emph{Stein operators}, 
kernel Stein discrepancy (KSD) \citep{gorham2017measuring} has been developed 
for non-parametric goodness-of-fit testing procedures for \emph{unnormalised models}
and demonstrated superior test performances in various scenarios including Euclidean data $\R^d$ \citep{chwialkowski2016kernel,liu2016kernelized}, discrete data \citep{yang2018goodness}, point processes \citep{yang2019stein}, latent variable models \citep{kanagawa2019kernel}, conditional densities \citep{jitkrittum2020testing}, censored-data \citep{fernandez2020kernelized}, directional data \citep{xu2020stein}, and network data \citep{xu2021graph}. 
It is worth to mention that previous KSD-based goodness-of-fit tests require developing case-specific Stein operators 
to address  
their corresponding data scenarios. The Stein operators can have diverse forms and may seem to be unrelated.


To improve the test performances, existing works have focused on data-adaptive methods for kernel learning \citep{gretton2012optimal,sutherland2016generative,liu2020learning} and kernel selection \citep{kubler2020learning, lim2020kernel}.
In addition, using the techniques related to kernel mean embedding \citep{muandet2017kernel}, KSD-based tests also enable the extraction of distributional features to perform computationally efficient tests and model criticisms \citep{jitkrittum2017kernel, xu2021stein}.
Nonetheless, the choice of Stein operators can also have a considerable effect for test performances but this aspect of the research 
has 
so far been ignored mostly. 
Previous works have only demonstrated the non-uniqueness of valid Stein operators, e.g. \cite{yang2018goodness} mentioned in their Section 3.2 that for random
graph models the Stein operator can be built from indicator functions or normalised Laplacians.
Moreover, 
\cite{fernandez2020kernelized} also derived three distinct Stein operators for censored data based on orthogonal properties from survival analysis where their test performances are only compared empirically.
The main contribution of this paper is threefold.

\textbf{Unifying KSD-based Tests}
We first
propose a simple-enough framework that provides a unifying view for various KSD-based tests. 
We explicitly show how the seemingly-unrelated Stein operators for different testing scenarios are connected, via \emph{auxiliary functions} that will be defined and explained later in Section \ref{sec:GKSD}. 

\textbf{Comparing KSD-based Tests}
With the unifying framework, 
we then interpret our Stein operators via \emph{auxiliary functions} and compare KSD tests using optimality conditions to be introduced in Section \ref{sec:optimal_condition}. 

\textbf{Guiding New KSD-based Tests}
Moreover, we provide a systematic approach to develop KSD-based tests on intended testing scenarios with application on domain constraint densities in Section \ref{sec:gksd_application}. We derive the bounded-domain KSD (bd-KSD) and its corresponding statistical properties. 
We experimentally demonstrate the test performances with different operator choices on truncated distributions and compositional data, i.e. data defined on a simplex/hypersimplex.

We begin our discussion by introducing relevant existing Stein operators and KSD-based tests in Section~\ref{sec:prelim} and conclude 
with future directions in Section~\ref{sec:conclusion}.

\section{PRELIMINARIES}
\label{sec:prelim}
We first review a set of existing Stein operators developed in various data scenarios, followed by a brief reminder on KSD-based goodness-of-fit testing procedures. We also introduce the notion of It\^{o}'s process for making connections between the generator approach and score-based Stein operators. 

\subsection{Stein Operators and Stein's Method}\label{sec:stein_operators}
Given a distribution $q$, an operator
$\T_q$ is called a Stein operator w.r.t. $q$ if the following Stein's identity holds for \emph{test function} $f$: 
$    {\E}_q [\mathcal{T}_q {f}]=0.$
The class of such test function $f$ is referred to as the Stein class. 

\textbf{Euclidean Stein operator on $\R^d$}
The Stein operator for continuous densities in 
$\R^d$ with Cartesian coordinate has been previously introduced \citep{gorham2015measuring, ley2017stein}, which is also referred to as the Langevin-diffusion Stein operator \citep{barp2019minimum}\footnote{Another important approach to develop Stein operator is via the (infinitesimal) generator \citep{barbour1988stein}. The connection can be established via It\^{o}'s process in Section \ref{sec:ito}.}.
Let $\X = \R^d$ and $f_i:\X \to \R$ for $i=1,\dots,d$ be scalar-valued functions on $\X$. $\mathbf{f}(x) = (f_1(x), \dots, f_d(x))^{\top} \in \R^d$ defines a vector-valued function. 
Let $q$ be a smooth probability density on $\X$ that vanishes at infinity.
For a bounded smooth function $\mathbf{f}:\mathbb{R}^d \to \mathbb{R}^d$, the Stein operator $\mathcal{T}_q$ is defined by 
\begin{align}
\mathcal{T}_q \f(x) &= \f(x)^{\top} \nabla \log q(x)+\nabla \cdot \f(x).
\label{eq:steinRd}
\end{align}
Stein's identity holds for $\T_q$ in Eq.~\eqref{eq:steinRd} 
due to the integration by parts on $\mathbb{R}^d$:
$    {\E}_q [\mathcal{T}_q \mathbf{f}] = \int_{\R^d} \mathcal{T}_q \mathbf{f}(x) dq(x) = \sum_i \int_{\R^d} \frac{\partial}{\partial x^i} \left(f_i(x) q(x) \right) dx = 0,$
where the last equality holds since $f_i(x)$ is bounded and $ q(x)$ vanishes at infinity. 
Since the Stein operator $\mathcal{T}_q$ depends on the density only through the derivatives of $\log q$, it does not involve the normalisation constant of $q$: a useful property dealing with unnormalised models.

\textbf{Censored-data Stein operator}
In practical data scenarios such as medical trials or e-commerce, we encounter data with censoring where the actual event time of interest (or
\emph{survival times}) is not accessible but, instead, a \emph{bound} or \emph{interval}, in which the event time is known to belong, is observed. \cite{fernandez2020kernelized} proposed a set of Stein operators for
right-censored data, where the lower bound of the event time is observed. 
%
The right-censored data is observed in the form of $(t_i,\delta_i)$ where for the survival time $x_i$ and censoring time $c_i$, the observation time is $t_i = \min\{x_i,c_i\}$ and  $\delta_i=\ind_{\{x_i\leq c_i\}}$
 indicates if we are observing $x_i$ or $c_i$.
Denote $\mu_0$ as the density of event time $x$; $S_C$ as the survival function\footnote{Survival function 
is defined as $S(x) = 1-F(x)$ where $F(x) = \int_0^x \mu(s)ds$ denotes the c.d.f. of the event time.
} 
of the censoring time $C$; the test function $\omega: \R_{+}\to \R$ assumed to vanish at origin, i.e. $\omega(0)=0$; $\Omega$ the set of functions $\R_+\times\{0,1\}\to\R$, and the operator $(\T_0 \omega) \in \Omega$.
The censored-data Stein operator is defined as 
\begin{align}\label{eqn:survival_stein_operator}
    (\mathcal T_0 \omega)(x,\delta) = \delta \frac{\left(\omega(x)S_C(x)\mu_0(x)\right)'}{S_C(x)\mu_0(x)}. 
\end{align}
Denote $\E_0$ as taking expectation w.r.t. the observation pair $(x,\delta)$ where $x\sim \mu_0$, which needs to distinguish from {$\E_{\mu_0}$ that takes expectation over $x\sim \mu_0$}. Note that the censoring distribution, e.g. $S_C$ remains \emph{unknown}. 
By Eq.~\eqref{eqn:expectedPhi0} in Appendix~\ref{app:knwon_identity}, we have Stein's identity $\E_{0}[(\mathcal T_0 \omega)(x,\delta)] = 0$.

The key challenge for this Stein operator is that the survival function for censoring time $S_C$ is \emph{unknown} and \emph{not included} in the null hypothesis. Hence, \cite{fernandez2020kernelized} applied identities in survival analysis to derive a computationally feasible operator: the survival Stein operator. This operator have an unbiased estimation from the empirical observations. For the hazard function $\lambda_0$ associated with $\mu_0$, the survival Stein operator  $\mathcal{T}_0^{(s)}$ is defined as
\begin{align}
(\mathcal{T}_0^{(s)}{\omega})(x,\delta)&=
\delta \left(\omega'(x)+\frac{\lambda_0'(x)}{\lambda_0(x)}\omega(x)\right) -\lambda_0(x)\omega(x)
.\label{eqn:defiT0S}
\end{align}
By the martingale identities,
\cite{fernandez2020kernelized} also proposed the martingale Stein operator,
\begin{align}
 (\mathcal T_0^{(m)}\omega)(x,\delta) = \delta\frac{\omega'(x)}{\lambda_0(x)}-\omega(x)
 .   \label{eqn:defiT0M}
\end{align}
Details for known identities regarding survival analysis and martingales can be found in Appendix~\ref{app:knwon_identity}.

\textbf{Latent-variable Stein operator}
Latent variable models are powerful tools in generative modelling and statistical inference. However, such models generally do not have closed form density expressions due to the integral operation on latent spaces.
Latent-variable Stein operator \citep{kanagawa2019kernel} is constructed via sampling the latent variables and the corresponding conditional densities . 
Let $q(x)\propto \int q(x|z)\pi(z)dz$ be the target distribution which is not accessible in closed form, even its unnormalised version. Sample $z_1,\dots,z_m \sim q(z|x)$. The latent variable Stein operator is defined as
\begin{equation}\label{eq:latent_stein_operator}
    \T_{q,z} \f (x) = \frac{1}{m} \sum_{j=1}^m \T_{q(x|z_j)} \f(x)
\end{equation}
A closely related construction is the \emph{Stochastic Stein Operator} \citep{gorham2020stochastic}, which has been developed for computationally efficient posterior sampling.
Additional details and comparisons with latent variable Stein operator are included in Appendix~\ref{app:ssd}.

\textbf{Second-order Stein operator}
To respect the geometry for distributions defined on Riemannian manifolds, Stein operators involving second-order differential operators have been studied \citep{barp2018riemannian, le2020diffusion}. For a smooth Riemannian manifold $\M$ and scalar-valued function $\tilde{f}:\M \to \R$, the second-order Stein operator 
is 
\begin{equation}
    \T^{(2)}_q \tilde{f}(x) =  \nabla \tilde{f} (x) ^{\top} \nabla \log q + \Delta \tilde{f}(x), \quad x \in \M,
\label{eq:second_order_operator}
\end{equation}
where $\Delta \tilde{f} = \nabla \cdot \nabla \tilde{f}$ denotes the corresponding Laplace-Beltrami operator\footnote{
Specifically, consider $\partial x^i$ as the basis vector on Tangent space of $x$, $\rm T_x \M$. 
 $\nabla \tilde{f} = \sum_{i,j} [G^{-1}]_{i,j} \frac{\partial \tilde{f}}{\partial x^j} \partial x^i,$ denotes the (Riemannian) gradient operator,
 where $G \in \R^{d\times d}$ denotes the metric tensor matrix;
 the divergence operator is 
$\nabla \cdot \mathbf{s} = \sum_i \frac{\partial s_i}{\partial x^i} + s_i \frac{\partial}{\partial x^i}\log\sqrt{\operatorname{det}(G)}$ for $\mathbf{s}=s_1\partial x^1,\dots,s_d\partial x^d$. In the Euclidean case, $G\equiv I$, the identity matrix, which is independent of $x$.}. 
This second-order operator can be seen as a natural consequence of diffusion process in $\M$ \citep{le2020diffusion}. Euclidean manifold is a special case of $\M$, on which the diffusion process would induce a similar second-order Stein operator: an example shown in Section \ref{sec:ito}. In the Euclidean case, another connection with
$\T_q$
in Eq.~\eqref{eq:steinRd} is via choosing test function in particular form: $\f = \nabla \tilde{f}$.

\textbf{Coordinate-dependent Stein operator}
Consider
coordinate system $(\theta^1, \dots, \theta^d)$ that is almost everywhere in $\M$. For a density $q$ on $\M$,
the Stein operator with the chosen coordinate is defined as
\begin{equation}\label{eq:stein_directional}
{\mathcal{T}}^{(1)}_{q} \mathbf f = \sum_{i=1}^{d} \left( \frac{\partial f^i}{\partial {\theta}^i} + f^i \frac{\partial}{\partial {\theta}^i} \log (qJ) \right),
\end{equation}
where $J = (\det G)^{1/2}$ is the volume element. ${\mathcal{T}}^{(1)}_q$ can be shown to satisfy Stein's identity using differential forms and the corresponding Stoke's theorem \citep{xu2020stein}.

\paragraph{Stein's method and standardisation techniques}
Stein's method \citep{stein1972bound} refers to the characterisation related to Stein operators for analysing distributions properties or distribution approximation \citep{barbour2005introduction, barbour2005multivariate, barbour2018multivariate} via solving the following Stein equation:
\begin{equation}\label{eq:stein_equation}
    \T_q f_h (x) = h(x) - \E_q [h(x)]
\end{equation}
where the subscript $h$ of $f$ explicitly refers to that function $f$ is associated with $h$ under regularity conditions. 
$\T_q$ denotes the corresponding Stein operator. Using centered Gaussian distribution in $\R^d$ as an example, $p(x)\propto \exp\{-\frac{1}{2} x^{\top} \Sigma^{-1} x\}$ where $\Sigma$ denotes the covariance matrix. The Stein operator derived from Eq.~\eqref{eq:steinRd} has the form
\begin{equation}\label{eq:stein_standardised}
    \mathcal{T}_q \f(x) = -  (\Sigma^{-1} x)^{\top} \f(x)+\nabla \cdot \f(x).
\end{equation}

However, during analysis and computation, the inverse covariance (or the precision) matrix can cause potential issues. The standaradisation techniques have been studied \citep{ley2017stein, mijoule2021stein}. For Gaussian distributions with constant convariance matrix, pre-multiply $\Sigma$ on the Stein operator would not destroy Stein's identity and the Stein operator reads 
$-x^{\top} \f(x)+\nabla \cdot \Sigma \f(x)$ that can be analysed nicely. 
The score function that interacts with $\f$ in $\T_q$ now becomes the score  of \textbf{standard} Gaussian.
Beyond Gaussian case, general standardisation Stein operators has been recently studied \citep[Definition 2.6]{ernst2020first}.

\subsection{
KSD Tests for Goodness-of-fit}\label{sec:ksd_intro}
With any well-defined Stein operator,
we can choose an appropriate RKHS w.r.t. the data scenario to construct its corresponding KSD.
Let $p$, $q$ be distributions satisfying regularity conditions for the relevant testing scenarios
and the test function class to be the unit ball RKHS, $B_1(\H)$, 
KSD between distributions $p$ and $q$ is defined as 
\begin{equation}
\operatorname{KSD}(p \| q; \H) 
= \sup_{\mathbf{f} \in B_1({\H}) }\E_{p}[\mathcal{T}_q \mathbf{f}] .
\label{eq:ksd}
\end{equation}
It is known from Stein's identity that for any test functions in the Stein class, $p=q$ implies $\operatorname{KSD}(p\|q ) = 0$.
In the testing procedure, a desirable property of the discrepancy measure is that $\operatorname{KSD}(q\|p)=0$ if and only if $p=q$. As such, we require our RKHS to be sufficiently large to capture any possible discrepancies between $p$ and $q$, which requires mild regularity conditions \citep[Theorem 2.2]{chwialkowski2016kernel} for KSD to be a proper discrepancy measure.
%
Algebraic manipulations produce the following quadratic form:
\begin{align}
\operatorname{KSD}^2(p \| q) = \E_{x,\tilde{x} \sim p} [h_q(x,\tilde{x})], \label{eq:KSDequiv}
\end{align}
where $h_q(x,\tilde{x}) = \langle \mathcal{T}_q k(x,\cdot), \mathcal{T}_q k(\tilde{x},\cdot)\rangle_{\H}$ does not involve  $p$; $k(x,\cdot)$ denotes the kernel associated with RKHS $\H$.

\paragraph{Testing procedure}
Now, suppose we have relevant samples $x_1,\dots,x_n$ from the \emph{unknown} distribution $p$. 
To test the null hypothesis $\mathrm{H_0}: p=q$ against the (broad class of) alternative hypothesis $\mathrm{H_1}: p\neq q$,
KSD can be empirically estimated via Eq.~\eqref{eq:KSDequiv} using U-statistics or V-statistics \citep{van2000asymptotic};
given the significance level of the test, the critical value can be determined by wild-bootstrap procedures \citep{chwialkowski2014wild} 
or spectral estimation \citep{jitkrittum2017linear} based on the Stein kernel matrix $H_{rs} = h_q(x_r,x_s)$; 
the rejection decision is then made by comparing empirical test statistics with the critical value.
In this way, a systematic non-parametric goodness-of-fit testing  procedure  is obtained, which is applicable to unnormalised models.

\subsection{It\^{o}'s Process and Infinitesimal Generator}\label{sec:ito}
It\^{o}'s process is fundamentally studied in stochastic differential equation (SDE) and can be described via continuous-time Markov process of the following form,
\begin{equation}\label{eq:ito}
    dX_t = b(X_t) dt + \sqrt{2}\sigma(X_t) d\mathbf{B}_t,
\end{equation}
where $b(X_t)$, $\sigma(X_t)$ denote the drift and diffusion functions of the process respectively; $\mathbf{B}_t$ denotes the standard Brownian motion. Under regularity conditions, the process reaches stationary distribution $\pi(x)$ for $b(x) = - (\log \pi(x))' \sigma(x)^2 + 2\sigma(x) \sigma(x)'$, where $'$ denotes the derivative \citep{stramer1999langevin}. For $\sigma \equiv 1$, the process is referred to as Langevin-diffusion. 
Vector-valued diffusion can be also established \citep[Definition 1]{gorham2019measuring} \citep[Example 3.20]{mijoule2021stein}.

If $(X_t)_{t\geq 0}$ is Feller process with invariant measure
$\pi$,
the (infinitesimal) generator for the process \citep[Section 2.2.1]{anastasiou2021stein} is defined (pointwise) as
\begin{equation}\label{eq:generator}
    (\A f) (x) = \lim_{t\downarrow 0} \E{[f(X_t)|X_0=x]} - f(x).
\end{equation}
It is not hard to see Stein's identity hold $\E_\pi[(\A f) (x)] = 0$ where the generator is a valid Stein operator. Moreover, if the test function $f$ is twice differentiable, the generator for the process in Eq.~\eqref{eq:ito} admits the following form
\begin{equation}\label{eq:ito_generator}
    (\A f) (x) = b(x) f'(x) + \sigma(x)^2 f''(x).
\end{equation}
For Langevin-diffusion where $\sigma(x)\equiv 1$,
the generator for the process with stationary distribution $q$ becomes
$(\A_q f)(x) = \log q(x)'f'(x) + f''(x)$, which retrieves the second order Stein operator $\T_q^{(2)}$ in Eq.~\eqref{eq:second_order_operator} for $\X = \R$.


\section{STANDARDISATION-FUNCTION KERNEL STEIN DISCREPANCY}\label{sec:GKSD}

\subsection{Standardisation-function
Stein Operator}
We introduce a simple-enough framework that is capable to unify the Stein operators introduced in 
Section \ref{sec:stein_operators}.
The idea is inspired from the standardisation technique for Eq.~\eqref{eq:stein_standardised}. We note that standardisation for analysing the Gaussian case in Eq.~\eqref{eq:stein_standardised} is done by pre-multiplying a \emph{constant} covariance function, which only have scaling effect on the Stein operator\footnote{
In the form of  
Eq.~\eqref{eq:ito}, constant diffusion 
$\sigma(x) \equiv c$
applies.}. 
Here, we consider the Stein operator standardised by matrix-valued \emph{auxiliary function} $G(x):\X \to \R^{d\times d}$.

Let $\f$ be an appropriately test function as defined in Section \ref{sec:stein_operators}. We consider the following Stein operator for density $q$, 
\begin{equation}\label{eq:sf_operator}
    \T_{q, G} \f =
\mathcal{T}_q (G\f) = (G\f)^{\top} \nabla \log q+\nabla \cdot (G\f),
\end{equation}
where $(G\f)\in \R^d$ stands for the matrix multiplication of $G(x)$ with $\f(x)$. We note this operator is analogous to the diffusion Stein operator \citep{barp2019minimum, mijoule2021stein} studied in the estimation context.
For our unifying purpose, it is enough to reduce $G(x)$ as a diagonal form with diagonal entries $\g=(g_1,\dots,g_d): \R^d \to \R^d$ with corresponding 
 $\T_{q, \g}$ being
\begin{equation}\label{eq:general_stein_operator}
(\T_{q, \g} \f) 
=\sum_{i=1}^d g_i (f_i \frac{\partial}{\partial x^i} \log q 
+ \frac{\partial}{\partial x^i} f_i)
+ f_i \frac{\partial}{\partial x^i}g_i.
\end{equation}
Stein's identity holds $\E_q [(\T_{q,\g} \f)(x)] = 0$ holds for all bounded function $\g$, due to similar integration by parts argument used for Eq.~\eqref{eq:steinRd}.
With Stein operator  $\T_{q, \g}$, 
we define the standardisation-function kernel Stein discrepancy,
\begin{equation}\label{eq:gksd}
    \operatorname{Sf-KSD}_{\g}(p \|q; \H) = \sup_{\|\f\|_{\H} \leq 1} \E_p[\T_{q,\g} \f].
\end{equation}
$\operatorname{Sf-KSD}_{\g}^2$ admits the quadratic form similar to Eq.~\eqref{eq:KSDequiv}.
\begin{Proposition}\label{prop:gksd_quadratic_form}
Let $\H$ be RKHS associated with kernel $K$. For fixed choice of bounded function 
$\g$, 
\begin{equation}\label{eq:gksd_quadratic}
\operatorname{Sf-KSD}^2_{\g}(p \|q; \H) = \E_{x,\tilde x \sim p}[h_{q,\g} (x,\tilde x)],
\end{equation}
where $h_{q,\g} (x,\tilde x) =  \left\langle \mathcal{T}_{q,\g} K(x,\cdot), \mathcal{T}_{q,\g} K(\tilde{x},\cdot)\right\rangle_{\H}$.
\end{Proposition}

For different choice of auxiliary functions $\g$, Sf-KSD exhibits distinct diffusion pattern induced by $\g$, i.e. $\g$ corresponds to the diagonal of second moment for diffusion $\sigma(x)\sigma(x)^{\top}$. We note that, by choosing $g_i(x) \equiv 1, \forall i \in [d]$, Sf-KSD recovers the KSD with Stein operator in $\R^d$ in Eq.~\eqref{eq:steinRd}, induced by Langevin diffusion in form of Eq.~\eqref{eq:ito}. Related ideas have been discussed using interpretations on Fisher information metric \citep[Section 7.2]{mijoule2018stein}; as well as the Bregman divergence induced mapping
in mirrored descent Stein sampler \citep{shi2021sampling}.
Analogous ideas also appeared in Stein variational gradient descent (SVGD)
\citep{liu2016stein}
for discrete variables \citep{han2020stein}.
We now proceed to show connections for existing KSDs proposed for 
various goodness-of-fit testing scenarios using Sf-KSD. 

\subsection{Unifying Existing Stein Operators}   
The specific choice of $\g$ and its interplay with $\f$ can be helpful to understand the conditions in various testing scenarios. 
With appropriate choice of the auxiliary function $\g$,
Sf-KSD is capable to provide a unifying view for the KSDs
introduced in Section \ref{sec:stein_operators}. To specify the equivalence notion, we denote $\triangleq$ as identical formulations 
beyond the equality in evaluations. 
All proofs and derivation details are included in the Appendix~\ref{app:proof}.

\begin{Theorem}[Censored-data Stein operator]\label{thm:cksd_equivalent}
Let dimension of the data $d=1$ and w.l.o.g., the test function $\omega$ 
vanishes at $0$. For $g:\R_+ \to \R$, choosing $g(x)=S_C(x)$, $\GKSD$ with $\T_{\mu_0,g}$ recovers the censored-data $\KSD$ with Stein operator defined in Eq.~\eqref{eqn:survival_stein_operator},
\begin{equation}\label{eq:survival_stein_equation}
\E_{\mu_0}[\T_{\mu_0,S_C}\omega] \triangleq \E_{0}[(\mathcal T_0 \omega)(x,\delta)] = 0.
\end{equation}
\end{Theorem}
It is not difficult to see that the result holds from directly applying identity in Eq.~\eqref{eqn:expectedPhi0} (explained in Appendix~\ref{app:knwon_identity}). However, it is worth noting that during the testing procedure,
$S_C$ is \emph{unknown} so we do not have direct access to $g$ here. 
Moreover, the expectation on l.h.s. of Eq.~\eqref{eq:survival_stein_equation} is w.r.t. the density of survival time $\mu_0$ for Sf-KSD, where the expectation on the r.h.s. is $\E_0$, w.r.t. the paired observation incorporating censoring information. Theorem \ref{thm:cksd_equivalent} serves the purpose of explicitly demonstrating how a particular choice of auxiliary function can bridge the gap between censored-data Stein operator with the Stein operator on 
distributions without the presence censoring information.
It will also be interesting to understand how the auxiliary function may explain the Stein operators in Eq.~\eqref{eqn:defiT0S} and Eq.~\eqref{eqn:defiT0M}  when $\E_0$ applies to Sf-KSD when the expectation is taken over the paired variable $(x,\delta)$.

\begin{Theorem}[Martingale Stein operator]\label{thm:mksd_equivalent}
Assume the same setting as in Theorem \ref{thm:cksd_equivalent}. Further assume that the positive definite test function $\omega:\R_+\to \R$ is integrable such that $ \int_0^x \omega(s)ds<\infty, \forall x$; $\mu_0(x) > 0$ for the survival times so the inverse of its corresponding hazard function $\lambda_0(x)=\frac{\mu_0(x)}{S_0(x)}$ is then well-defined on $\R_+$.
For $g:\R_+ \to \R$, choosing   
$g(x)=\delta \lambda_0(x)^{-1} + (1-\delta)\frac{\int_0^x \mu_0(s)\omega(s)ds}{\mu_0(x)\omega(x)}$, $\GKSD$ with $\T_{\mu_0,g}$ recovers the martingale $\KSD$ with Stein operator defined in Eq.~\eqref{eqn:defiT0M},
\begin{equation}
\E_{0}[\T_{\mu_0,g}\omega] \triangleq \E_{0}[(\mathcal T^{(m)}_0 \omega)(x,\delta)] = 0.
\label{eq:martingale_stein_equation}
\end{equation}
\end{Theorem}
The $\delta$-dependent decomposition of $g$ in  Theorem \ref{thm:mksd_equivalent} reveals the relationship between how censoring is incorporated in the martingale Stein operator, i.e. through the hazard function for uncensored data while through an interaction between the density $\mu_0$ and the test function $\omega$ in the censored part.
Similarly, choosing $\delta$-dependent auxiliary function $g$ can recover survival Stein operator in Eq.~\eqref{eqn:defiT0S}.

\begin{Corollary}[Survival Stein operator]\label{cor:sksd_equivalent}
Assume the conditions in Theorem~\ref{thm:mksd_equivalent} hold. 
Consider $\zeta(x) = \frac{\int_0^x \mu_0(s)\omega(s)\lambda_0(s)ds}{\mu_0(x)\omega(x)}$.
For $g:\R_+ \to \R$, choosing $g(x)=\delta + (1-\delta)\zeta(x)$, $\GKSD$ with $\T_{\mu_0,g}$ recovers the survival $\KSD$ with Stein operator defined in Eq.~\eqref{eqn:defiT0S},
$\E_{0}[\T_{\mu_0,g}\omega] \triangleq \E_{0}[(\mathcal T^{(s)}_0 \omega)(x,\delta)] = 0.$
\end{Corollary}

\paragraph{{Comparisons between $\mathcal T^{(m)}_0 $ and $\mathcal T^{(s)}_0 $}} 
From 
Theorem \ref{thm:mksd_equivalent} and Corollary \ref{cor:sksd_equivalent} above, we are able to explicitly see the following.

1) \textit{in the uncensored part:} the diffusion for martingale Stein operator $\T^{(m)}_0$ is through the inverse of hazard function while the survival Stein operator $\T^{(s)}_0$ has constant auxiliary function, replicating constant diffusion in Eq.~\eqref{eq:ito}.

2) \textit{in the censored part:}
both Stein operators rely on the integral form where density $\mu_0$ and test function $\omega$ interacts, while 
$\T^{(s)}_0$ involves the hazard function $\lambda_0$ within the integral that can be harder to estimate empirically. 

Our results show that for $\mathcal T^{(s)}_0 $, the censoring information is only 
extracted from the censored part of data ($\delta = 0$) while
the uncensored part of data  ($\delta = 1$)
are treated exactly the same as 
Eq.~\eqref{eq:steinRd}.
However for $\mathcal T^{(m)}_0$, the censoring information is re-weighted (or re-scaled) via both censored part and uncensored part of data,
which results in more accurate empirical estimation compared to that of $\mathcal T^{(s)}_0 $. The theoretical interpretations corroborate the empirical findings reported in \cite{fernandez2020kernelized}. 
Additional comparisons and interpretation based on optimality conditions for KSD test in Section \ref{sec:optimal_condition} are detailed in Appendix~\ref{supp:comparison}.

\begin{Theorem}[Latent-variable Stein operator]\label{thm:ssd_equivalent}
Assume  $q(x|z)$ vanishes at infinity $\forall z$. Given sample $z=\{z_j\}_{j\in [m]} \sim q(z|x)$, the $z$-dependent 
function $\g^z : \R^d \to \R^d$ is  chosen as $g^z_i(x)=\sum_j \delta_{z_j}(x,z)$\footnote{For test function $F(x,z)$, the bivariate delta function (of the second argument) satisfies $\int \delta_{z_j}(x,z)F(x,z)dz = F(x,z_j).$}. $\GKSD$ recovers the latent-variable Stein operator in Eq.~\eqref{eq:latent_stein_operator}, 
$\E_{q}[\T_{q,\g^z}\f] \triangleq \sum_j \E_{q(x|z_j)} [\T_{q(x|z_j) }\f] 
= 0.$
\end{Theorem}
By choosing the auxiliary function as the finite sum of delta measures on the latent variable locations, the auxiliary function is effectively performing the sampling procedure to construct the random kernel for
the latent-variable Stein operator in \cite{kanagawa2019kernel}, which surpasses the intractability arising from integral operation over the latent variables. 


\begin{Theorem}[Second-order Stein operator]\label{thm:second_order_equivalent}
Let $g(x)=\frac{\partial}{\partial x}\log f(x)$, $\GKSD$ recovers the  second-order Stein operator defined in Eq.~\eqref{eq:second_order_operator},
$\T_{q,\log f '} f \triangleq \T^{(2)}_q f.$
\end{Theorem}

Using the fact that $g \cdot f = (\log f)'f= f'$, having $g = (\log f)'$ produces the extra order on the differential operator. For the more general formulation
based on the linear operator are discussed in Appendix~\ref{app:linear_generalisation}. By choosing the linear operator itself to be the differential operator will automatically recover such second-order Stein operator. The multivariate version and the Riemannian manifold version are also applicable. Details are included in the Appendix \ref{app:linear_generalisation}.


\begin{Theorem}[Coordinate-dependent Stein operator]\label{thm:coordinate_equivalent}
Let $g_i(x)={\log J}, \forall i \in [d]$, $\GKSD$ recovers the Stein operator defined in Eq.~\eqref{eq:stein_directional},
$$
\sum_i \int_{\X} \T_q (f_i\log J) d qJ = \sum_i \int_{\X}
\frac{\partial}{\partial \theta^i}(f_i q J) =0.
$$
\end{Theorem}

{\it Proof. The result follows from separating the last term in Eq.~\eqref{eq:stein_directional}: $\log(qJ) = \log q + \log J$. $\square$}

With the particular choice of coordinate system, choosing the auxiliary function $g$ to be the log of Jacobian can be interpreted as changing the diffusion pattern to incorporate coercive expectation w.r.t. taking expectation over the density. 
This can be explicitly shown using Stoke's theorem with differential forms \citep{xu2020stein}.
In addition, the idea of using auxiliary function to incorporate domain properties or constraints can be very useful for problems where the data has a complicated and irregular domain. We provide case study on domain constraint distributions in Section \ref{sec:gksd_application}.

\section{OPTIMALITY CONDITIONS}\label{sec:optimal_condition}

We have analysed how the choice of fixed auxiliary function can provide a unifying view on KSD-tests in different testing scenarios. However, Stein operators for the same test scenario is non-unique and under what conditions the choice of Stein operators is the optimal?
As we are thinking of optimality conditions for the auxiliary functions, we tackle this point via a
\emph{calculus of variation} approach \citep{gelfand2000calculus}. 

Consider a cost function
$L(y,\dot{y})$ (usually in integral form) that depends on a path function $y(x)$ and its derivative $\dot{y}(x) = \frac{d}{dx} y(x)$. The calculus variation approach perturbs the path function $y$ by a minimal amount and ask when is $y$ a stationary path w.r.t.
$L(y,\dot{y})$. Then condition is then characterised via the following Euler-Lagrangian equation,
\begin{equation}\label{eq:euler-lagrangian}
    \frac{\partial}{\partial y}L(y(x),\dot{y}(x)) - \frac{d}{dx} \frac{\partial}{\partial \dot{y}}L(y(x),\dot{y}(x)) = 0.
\end{equation}
Note that the differential w.r.t $y$ is the path perturbation instead of moving variable $x$.
Denote $\dot{f}_i = \frac{\partial f_i}{\partial x^i} $ and $\dot{\f} = (\dot{f_1},\dots,\dot{f_d})$ (the same notion applies to $\g$).
The cost Sf-KSD has the path integral form $L(\f, \dot{\f}, \g, \dot{\g}; p, q) = \E_p[\T_{q,\g} \f]$ based on Eq.~\eqref{eq:general_stein_operator}. 
 
\begin{Theorem}[Optimality conditions]\label{thm:optimal}
$\GKSD$ in Eq.~\eqref{eq:gksd} admits a stationary point w.r.t. $\f$ when the following 
holds,
\begin{equation}\label{eq:optimality}
    \E_p[ g_i(x)  \frac{\partial}{\partial x^i} \log q(x)] = 0, \quad \forall i\in[d].
\end{equation}
\end{Theorem}
The proof follows from applying the Euler-Langrangian equation of Eq.~\eqref{eq:euler-lagrangian} for each $f_i$, $i\in[d]$.
As we choose $f_i$ to be in the same function space and only interact with $g_i$ due to Eq.~\eqref{eq:stein_standardised} for all $i$, the symmetry on optimality conditions on $g_i$ is expected.

\paragraph{Variance constraints and connections to the infinitesimal  generator}
The optimality condition in Eq.~\eqref{eq:optimality} holds with flexibility in scaling. Following the fashion where KSD in Eq.~\eqref{eq:ksd} uses unit ball RKHS constraint, we would like to regularise the diffusion variance w.r.t. data distribution $p$.
Making connection with diffusion process, its generator  in Eq.~\eqref{eq:ito_generator}
variance implies the diffusion function $\sigma(x)^2$ corresponds to $g(x)$. As such, we would like to regularise $\Var_p[dX_t] \leq 1$. As Sf-KSD is monotonic w.r.t. scaling, our variance regularisation becomes equality constraint $\E_p[g(x)] = 1$.

It is not hard to check that the density ratio function, or the importance weight, $g(x) = \frac{q(x)}{p(x)}$ satisfies both the optimality conditions in Eq.~\eqref{eq:optimality} and the variance constraint. 
However, in goodness-of-ft testing, we do not know data distribution $p$ (only accessible through sample form), we are unable to derive $g$ directly.
Density ratio and its estimation have been extensively studied in the literature \citep{hido2011statistical,kanamori2009least,sugiyama2012density} and its link with kernel methods was pioneered in \cite{kanamori2012statistical}. 
Regarding density ratio argument, it is also interesting to see that $g(x) \equiv 1$ satisfies the optimality condition when $p=q$, which corresponds to the null hypothesis.

\section{APPLICATION: TESTING WITH DOMAIN CONSTRAINTS}
\label{sec:gksd_application}
We show that different choices of auxiliary functions are able to induce appropriate Stein operators for various scenarios. Sf-KSD can then be useful to develop a systematic approach for new kernel Stein tests when appropriately auxiliary functions are used. 
In this section, we apply the Sf-KSD framework for testing data with general domain constraint. 


Let $q$ be a probability distribution defined on a compact domain\footnote{It is common that $V$ is embedded in some non-compact domain $\Omega$, e.g. truncated distribution from $\R^d$.} $V$ with boundary $\partial  V$. Denote the unnormalised density 
$\tilde q (x) \propto q(x), x \in V $.
Common examples 
include the truncated Gaussian distribution on the interval $[a,b]$ 
or compositional data that defined on a simplex/hyper-simplex. Complex boundaries such as polygon \citep{liu2019estimating, yu2020generalized} or non-negative constraint for graphical models \citep{yu2018graphical} have been studied. Such settings are common  when the observed data is only a subset of a larger domain or consists of structural constraint such as composition. For instance, if a local government would like to study the spread of the disease during pandemic if information about infections is not accessible from other countries, one may need to validate model assumptions with the domain truncated by the designated border.

\subsection{Stein Operators on Compact Domains}\label{sec:general_stein}
To create the KSD-type test for data on domain $V$, we first consider Stein operators for densities on $V$. 
We develop such a Stein operator guided by the Sf-Stein operator in Eq.~\eqref{eq:general_stein_operator}.
Unlike densities on unbounded domains that are commonly assumed to vanish at infinity, densities on compact domains may not usually vanish at the boundary. Hence, direct application of a Stein operator on $\R^d$ may require the knowledge of \emph{normalised} density at the boundary, which defeats the purpose of KSD testing for unnormalised models. To address this issue, 
%
%
we consider a bounded smooth function $\g: \R^d \to \R^d$ such that $g_i(\partial V) = 0, \forall i\in[d]$ and for unnormalised 
$\tilde{q}$ on $V$, the bounded-domain Stein operator is defined as 
$\widetilde{\T}_{ q,\g} \f (x)  = \left(\frac{1}{ q}\sum_i \frac{\partial}{\partial x^i} ( q g_i f_i)\right)(x)$.
With the aid from auxiliary function $\g$, it is not hard to check Stein's identity holds w.r.t. $q$.
%
The recent advances in sampling with domain constraint \citep{shi2021sampling} also utilise the Stein operator with mirror descent 
and perform SVGD \citep{liu2016stein} in $\psi$-transformed space. Based on the generator on similar It\^{o}'s process \citep[Theorem 10]{shi2021sampling}, their Mirrored Stein operator corresponds to $\T_{q,G}$ in Eq.~\eqref{eq:sf_operator} with $G = \nabla^2 \psi ^{-1}$. In simplex case, they choose $\psi$ to be negative-entropy to satisfy the boundary conditions. 

\subsection{Bounded-domain Kernel Stein Discrepancy}
With the Stein operator $\widetilde{\mathcal T}_{\tilde q, \g}$, we proceed to define the bounded-domain Kernel Stein Discrepancy (bd-KSD) for goodness-of-fit testing, similar to the Section \ref{sec:ksd_intro}. 
$\operatorname{bd-KSD}_{\g}(\tilde q\|\tilde p ) = \sup_{\f \in B_1(\H)} \mathbb{E}_{\tilde p}[\widetilde{\mathcal T}_{q,\g} \f(x)]$.
Standard reproducing property gives the
quadratic
form
$\operatorname{bd-KSD}_{\g}(\tilde q\|\tilde p )^2 = 
\E_{x,x' \sim \tilde q} [h_{q, \g}(x,{x}')]$,
where $h_{q, \g}(x,{x}') = \left\langle \widetilde{\mathcal T}_{q,\g}  k(x,\cdot),\widetilde{\mathcal T}_{q,\g} k({x}',\cdot)\right\rangle_{\H}$.
Let $L(x) = (L_1(x),\dots,L_d(x))^{\top}\in \R^d$ with
$   L_i(x)=\frac{\partial}{\partial x^i} \log \frac{q(x)}{p(x)},$
we show that under mild regularity conditions, bd-KSD is a proper discrepancy measure on $V$.


\paragraph{Goodness-of-fit testing with bd-KSD} 
Similar procedure as introduced in Section \ref{sec:ksd_intro} applies to test the null hypothesis ${\rm H}_0: \tilde p  = \tilde q$ against the alternative ${\rm H}_1: \tilde p  \neq \tilde q$. Observed samples $x'_1,\dots,x'_n \sim \tilde p$ on $V$, the empirical U-statistic \citep{lee90} can be computed,
$\operatorname{bd-KSD}_{\g}(\tilde q\|\tilde p )_u^2 = 
\frac{1}{n(n-1)} \sum_{i\neq j}[h_{q, \g}(x'_i,{x}'_j)]$.
The asymptotic distribution is obtained via U-statistics theory \citep{van2000asymptotic} as follows.  
We denote the convergence in distribution by $\overset{d}{\to}$.

\begin{Theorem}\label{thm:u-stat-asymptotic}
Consider U-statistic $\operatorname{bd-KSD}_{\g}(\tilde q\|\tilde p )_u^2$.
1) Under $H_0: \tilde p = \tilde q$, 
$n\cdot\operatorname{bd-KSD}_{\g}(\tilde q\|\tilde p )_u^2  \overset{d}{\to} \sum_{j=1}^{\infty} w_{j}(Z_{j}^{2}-1),$
where $Z_{j}$ are i.i.d. standard Gaussian random variables and $w_{j}$ are the eigenvalues of the Stein kernel $h_{q,\g}(x,{x}')$ under $\tilde p({x}')$:
$\int_V h_{q,\g}(x,{x}')\phi_{j}({x}')\tilde q ({x}'){\rm d}{x}' = w_{j}\phi_{j}(x)$,
where 
$\phi_{j}(x) \neq 0$ is the non-trivial eigen-function for Stein kernel operator $h_{q,\g}$.
2) 
Under $H_1: \tilde p\neq \tilde q$, 
\[
\sqrt{n}\cdot \left(\operatorname{bd-KSD}_{\g}(\tilde q\|\tilde p)_u^2 - \operatorname{bd-KSD}_{\g}(\tilde q\|\tilde p)^2\right)\overset{d}{\to}\mathcal{N}(0,{\sigma}^{2}_{H_1}),
\]
where ${\sigma}^{2}_{H_1}=\mathrm{Var}_{x\sim p}[\E_{\tilde{x}\sim p}[h_{q,\g}(x,{x}')]]>0$ produces the non-degenerate U-statistics. 
\end{Theorem}

The goodness-of-fit testing then follows the standard procedures in Section
\ref{sec:ksd_intro} by applying bd-KSD.


\subsection{Case Studies on Compact Domains
}\label{sec:exp}

We first consider 
\textbf{truncated distributions}. In Fig.~\ref{fig:demo_dist} (left), an example of two-components Gaussian mixture truncated in a unit ball is plotted in $B_1(\R^2)$. It is obvious that the density $q(x)$ does not necessarily vanish at the truncation boundary. 
Truncated distributions, including truncated Gaussian distributions \citep{horrace2005some,horrace2015moments}, truncated Pareto distributions \citep{aban2006parameter}, or truncated power-law distributions \citep{deluca2013fitting} have been studied. In particular, left-truncated distributions are of special interest in survival analysis \citep{klein2006survival}. 
To the best of our knowledge, goodness-of-fit testing procedures for general truncated distributions has not yet been established.

We also consider \textbf{compositional data} where distributions are defined on a simplex, 
${\rm{S}}^{d-1} = \{x^i\in[0,1], \sum_{i=1}^{d} x^i=1\},$ 
which is
a compact domain.  
A common example for compositional distribution is the Dirichlet distribution, with \emph{unnormalised} density of the form $\tilde q(x)\propto \Pi_{i=1}^d (x^i)^{\alpha_i-1}$, $\forall x\in {\rm{S}}^{d-1}$, where $\alpha_i>0$ are the \emph{concentration parameters}. 
An example Dirichlet distribution on ${\rm{S}}^{2}$ is illustrated in Fig.~\ref{fig:demo_dist} (right). 
It is also obvious that $q(x)$ does not necessarily vanish at the boundary\footnote{Specifically, $q(x)=0$ 
at boundary $D_i=\{x|x^i=0\}$ for $\alpha_i >1$; while $q(x)>0$ on $D_i$ for $\alpha_i \in (0,1]$.}.
Recently, score matching procedures have been proposed to estimate unnormalised models for compositional data \citep{scealy2020score}. To the best of our knowledge, goodness-of-fit testing procedures for general \emph{unnormalised} compositional distributions has not been well studied.

\begin{figure}[t!]
    \centering
    \includegraphics[width=0.21\textwidth,height=0.22\textwidth]{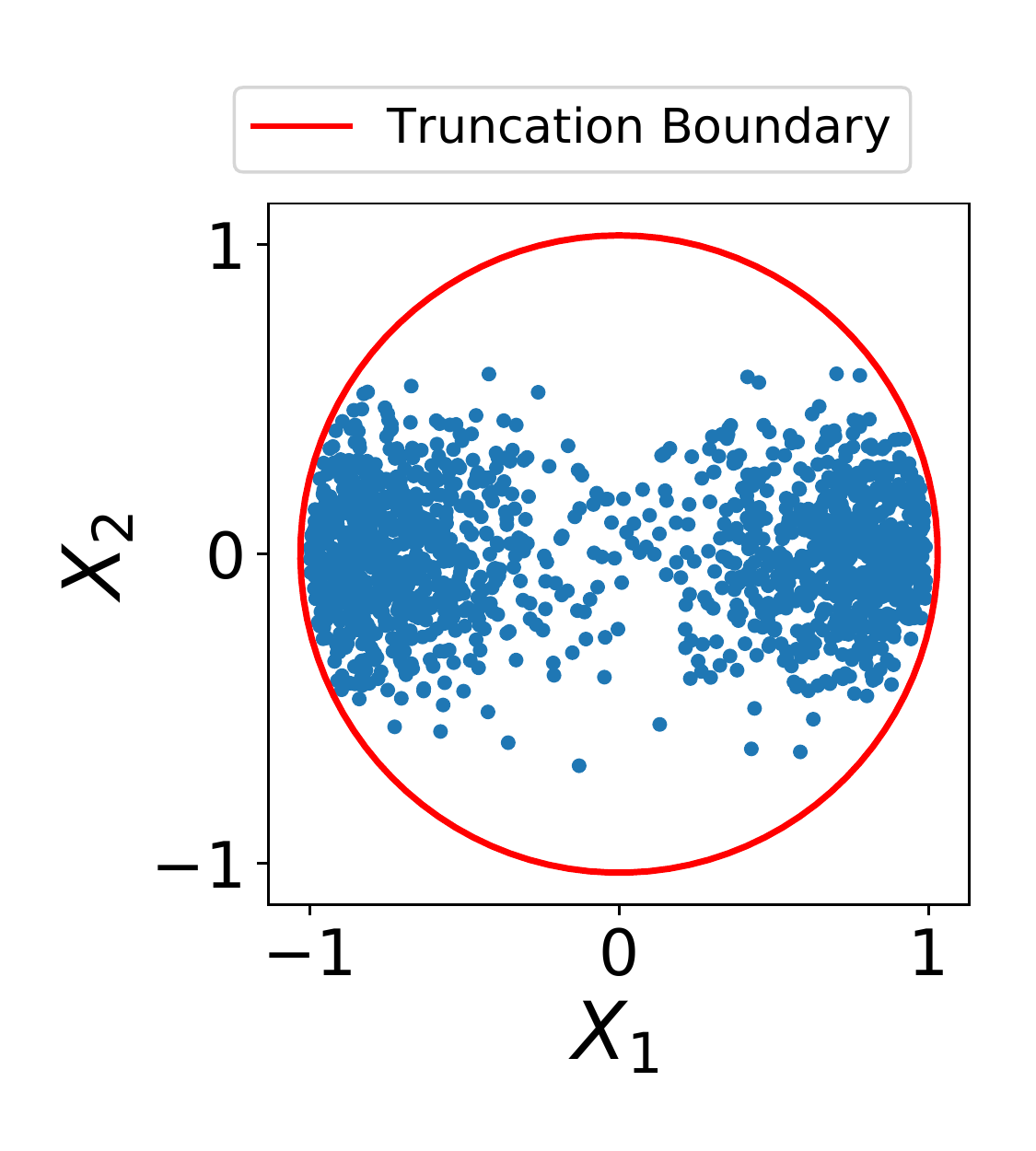}
    \includegraphics[width=0.25\textwidth,height=0.22\textwidth]{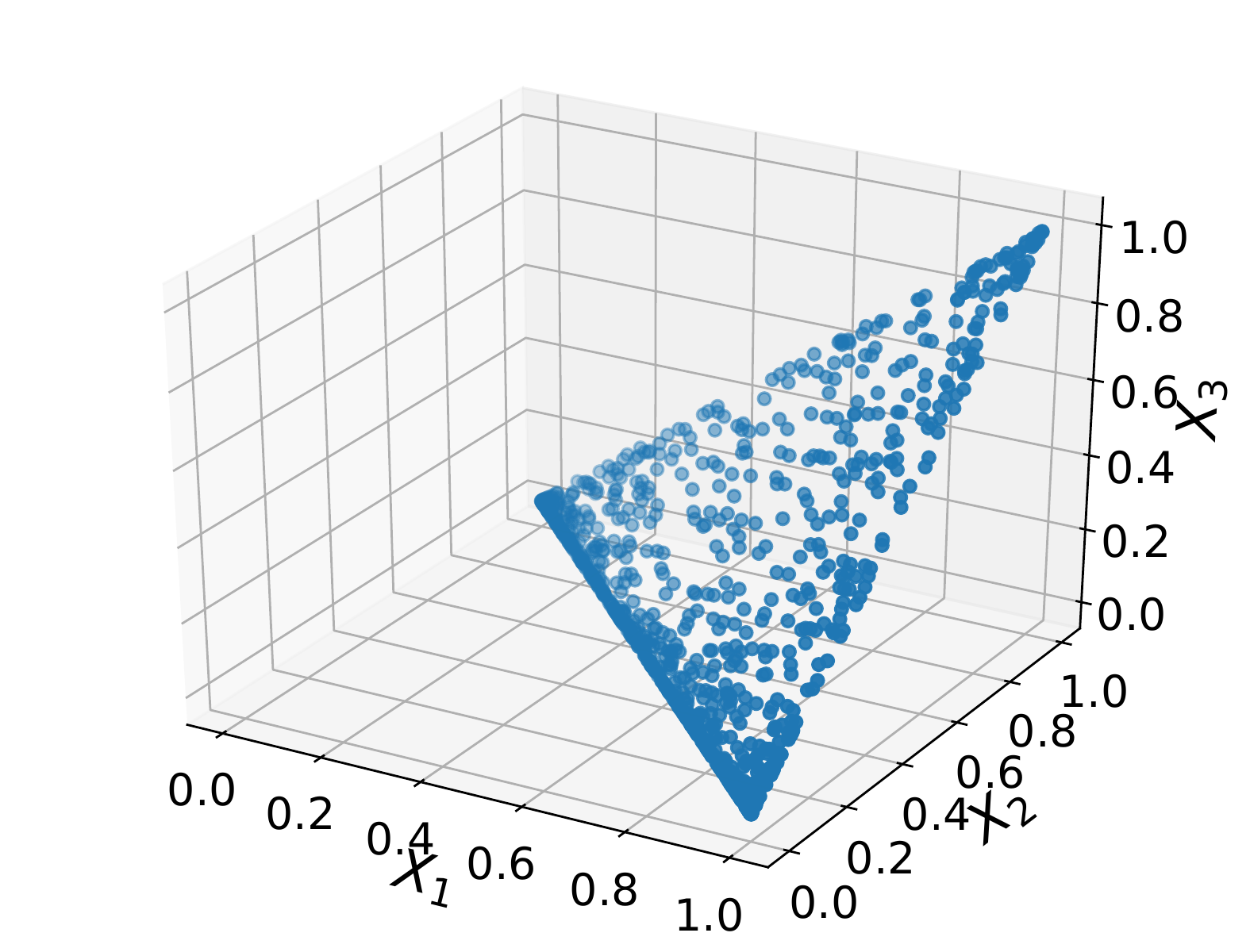}
    \caption{
    Sample distributions with compact domains. 
    Left: two-component Gaussian mixtures truncated within a unit circle;
Right: Dirichlet distribution on a simplex with $\boldsymbol{\alpha} = (0.2,0.5,0.5)$.}
    \label{fig:demo_dist}
    \vspace{-0.5cm}
\end{figure}

\textbf{Comparing tests for goodness-of-fit} 
An important aspect of applying bd-KSD is to choose the appropriate auxiliary functions in each data scenario, which we now specify. Simulation results are shown in Table.~\ref{tab:power_summary}. 
We compare the KSD-based approach to the maximum mean discrepancy (MMD) \citep{gretton2012kernel} based approach where samples are generated from the null model and a two-sample test is then performed. Such strategy to test goodness-of-fit has been considered previously \citep{
jitkrittum2017linear, xu2020stein}.
Another data-adaptive test based on kernel-embedding, named as Moderated MMD (M3D) \citep{balasubramanian2021optimality}, is also compared. Recently, \citet{shi2021sampling} has developed mirrored Stein discrepancy (MSD) for sampling purposes with applications for compositional data. We adapt and modify their formulation to perform goodness-of-fit testing procedure as comparison for our compositional data example, where the neg-entropy function is chosen as the mirrored map. 

\subsubsection{Simulation results on synthetic data}

\textbf{Truncated Distributions in Unit Ball $B_1(\R^d)$}

Requiring to vanish at the boundary and take into account the rotational invariance of the unit ball, the auxiliary functions can be chosen as $g_i^{(p)}(x) = 1 - \|x\|^p$, which relates to the Euclidean distance from the boundary raising to chosen power $p$. Similar form of auxiliary function, with $p=1$, was discussed for density estimation on truncated domains \citep{liu2019estimating}. 
For larger $p$, more weights are concentrated to the center of the ball. 
We present the case where the null is a two-component  Gaussian mixture with identity variances and the alternative is the two-component Gaussian mixture with correlation coefficient perturbed by $\nu$.
With such alternative distributions, $g^{(2)}$ is closer to the density ratio for optimality condition in Eq.~\eqref{eq:optimality}, making  bd-KSD($g^{(2)}$) a more powerful test as shown in Table.~\ref{tab:power_summary}, outperform MMD-based tests. M3D improves from MMD tests with better data-adaptive scheme, while is still outperformed by bd-KSD based tests.

\textbf{Compositional Distributions}
To vanish at the boundary
$D_i=\{x|x^i=0\}$, a natural choice of 
the auxiliary function 
is the  geometric mean function $g^{(1)}_i(x)=(\Pi_{i=1}^d x^i)^{1 / d}$. 
Moreover, minimum distance-to-boundary function $g_i^{(2)}(x) =\min\{\|x-z\| |z\in D_i \}$ satisfies the boundary conditions. We present Dirichlet distribution with concentration parameters $\alpha_i=0.5$ as the null and perturb $\alpha_1 = 0.5  + \nu$ as the alternative.
The geometric mean function $g^{(1)}$ is closer to the density ratio function, making it more
sensitive detect the difference on the boundary compared to
 the minimum distance-to-boundary function $g^{(2)}$,
bd-KSD($g^{(1)}$) produces higher power 
as shown in Table.~\ref{tab:power_summary}.
Moreover, results in Table.~\ref{tab:power_summary} also show that bd-KSD based tests outperforms the MMD based tests. Similar trends apply, M3D improves from MMD but is not more powerful than bd-KSD. MSD, being a specific form of Sf-KSD, performs competitively as compared to bd-KSD when the perturbation of the alternative becomes larger. 
%
We include additional details, simulation results and insights on $\g$ choices 
in Appendix~\ref{app:exp}.

\begin{table}[t]
  \centering
\vskip 0.1in
  \begin{tabular}{c|rrr}
\multicolumn{4}{c}{Truncated Data on $B_1(\R^3)$}\\
\toprule
  & \multicolumn{1}{l}{$\nu=0.1$} & \multicolumn{1}{l}{$\nu=0.3$} &\multicolumn{1}{l}{$\nu=1.0$} \\
    \midrule
    bd-KSD($g^{(1)}$) & 0.235& 0.760& {\bf 1.000}\\
    bd-KSD($g^{(2)}$)   & {\bf 0.305} & {\bf 0.910}& {\bf 1.000}\\
    MMD & 0.045& 0.210& {\bf 1.000} \\
    M3D & 0.120 & 0.310 & {\bf 1.000} \\
    \bottomrule
  \end{tabular}
\vskip 0.2in
\begin{tabular}{c|rrr}
\multicolumn{4}{c}{Compositional Data on ${\rm{S}}^{2}$}\\
\toprule
  & \multicolumn{1}{l}{$\nu=0.1$} & \multicolumn{1}{l}{$\nu=0.3$} &\multicolumn{1}{l}{$\nu=1.0$} \\
    \midrule
    bd-KSD($g^{(1)}$)       &{\bf 0.430} & {\bf 0.715}& {\bf 0.905}\\
    bd-KSD($g^{(2)}$)       & 0.325& 0.565 & 0.855\\
    MMD & 0.090 & 0.425 & 0.730 \\
    M3D & 0.135 & 0.505 & 0.805 \\
    MSD & 0.230 & 0.610 & 0.895\\
    \bottomrule
    \end{tabular}%
  \footnotesize
      \caption{Test power for simulation results. perturbed level $\nu$ for the alternatives; sample size $n=200$, test size $\alpha=0.01$; repeat $200$ trials. Bold number indicates the best power. }  \label{tab:power_summary}%
      
\end{table}%

\subsubsection{Real data experiments} 
We apply bd-KSD tests on two real datasets to assess the model goodness-of-fit. Models are fitted by density estimation techniques for these two data scenarios.

1. \emph{Chicago Crime Dataset}\footnote{Data source \url{https://data.cityofchicago.org}.} 

The dataset contains all crime locations within the city of Chicago, where we use ``robbery'' data in 2020. 
We consider  TruncSM \citep{liu2019estimating}, a score-matching based estimation objective, to fit a Gaussian mixture model\footnote{The model is fitted with half of the data and tested on the other half, which is referred to as ``data-splitting'' techniques \citep{jitkrittum2017linear,kubler2020learning}}. We set auxiliary function as the Euclidean distance to the nearest boundary point (analogous to $g^{(2)}$). For a 2-component Gaussian mixture, bd-KSD gives p-values 0.002 which is clearly an inadequate fit; for a 20-component Gaussian mixture, p-value is 0.162 which indicates a good fit of the TruncSM estimated model.

2. \emph{Three-composition AFM of 23 Aphyric Skye Lavas Dataset} 
\citep{aitchison1985kernel} 

The variables A, F and M represent the relative proportions of $Na_2+K_2O$, $Fe_2O_3$ and $MgO$, respectively. We fit the Gaussian kernel density estimation proposed by \citet{chacon2011asymptotics}, using half of the data and test on the other half. We choose the auxiliary function $g^{(2)}$: the min distance to the closest boundary. The bd-KSD gives p-value 0.004 which rejects the null hypothesis, indicating the fit is not good enough.


\section{CONCLUSION AND FUTURE DIRECTIONS}\label{sec:conclusion}
The present work studies a unifying framework Sf-KSD, to interpret and compare existing KSD-based  goodness-of-fit tests; as well as to design new tests. Optimality conditions for developing Sf-KSD are studied, making connections between the score-based Stein operator and the generator approach.
When performing goodness-of-fit tests, it is worth to note where the procedures may be mis-applied or wrongly interpreted in the wider scientific community where these must be guarded against. For instance, failure of p-value corrections from correlated samples can result in false positives, which is particularly risky in studies associated with healthcare.
With density ratio satisfying the optimality conditions, it is an interesting unexplored area that how density ratio estimation or importance reweighting can be helpful in learning KSD in the context of goodness-of-fit testing, which we leave as a future direction.


\subsubsection*{Acknowledgements}
W.X. acknowledges Gesine Reinert for the enlightenment on Stein's method and introduction to the standardisation techniques that inspired this work and many helpful discussions. W.X. would like to thank Arthur Gretton and Nan Lu for helpful discussions and constructive comments. W.X would also like to thank anonymous reviewers for helpful comments. W.X. is supported by the EPSRC grant EP/T018445/1.

\bibliographystyle{apalike}
\bibliography{stein}

\clearpage
\onecolumn
\appendix

\thispagestyle{empty}

\onecolumn \makesupplementtitle

    

\section{Known Identities}\label{app:knwon_identity}
\subsection*{Expectations in Survival Analysis}
We know the following identities in survival analysis, which will be useful for discussions in the main text: for any measurable function $\phi$,
\begin{align}
    \E_{0}
[\Delta \phi(T)]& = \int_0^{\infty} \phi(s)\mu_0(s)S_C(s)ds,\label{eqn:expectedPhi}\\
    \E_{0}
[(1-\Delta) \phi(T)]& = \int_0^{\infty} \phi(s)\mu_C(s)S_0(s)ds.\label{eqn:expectedPhi0}
\end{align}
where $\mu_C$ here denotes the p.d.f. of the censoring distribution and $S_0$ denotes the survival function w.r.t. $\mu_0$.

\subsection*{Martingales in Survial Analysis}
The following identity is useful to understand the martingale Stein operator in \cite{fernandez2020kernelized}
\begin{align}\label{eq:martingale_identity}
     \E_0\left[\Delta \phi(T)-\int_0^{T} \phi(t)\lambda_0(t)dt\right]=0,
 \end{align}
which holds under the null hypothesis, where $\lambda_0$ is the hazard function under the null $\mu_0$. 
Let $N_i(x)$ and $Y_i(x)$ be the individual counting and risk processes, defined by by $N_i(x)=\delta_i\ind_{\{T_i\leq x\}}$ and $Y_i(x)=\ind_{\{T_i\geq x\}}$, respectively.  Then, the individual zero-mean martingale for the i-th individual corresponds to $M_i(x)=N_i(x)-\int_0^xY_i(y)\lambda_0(y)dy$, where $\E_0[M_i(x)]=0$ for all $x$. 

Additionally, let $\phi:\R_+\to\R$ such that $\E_0\left|\int_0^x \phi(y) dM_i(y)\right|<\infty$ for all $x$, then $\int_0^x \phi(y) dM_i(y)$ is a zero-mean $(\mathcal{F}_x)$-martingale (see Chapter 2 of \citep{aalen2008survival}). Then, taking expectation, we have 
\begin{align*}
\E_0 \left[\int_0^\infty \phi(x)dM_i(x)\right] 
&= \E_0\left[\int_0^\infty \phi(x)(dN_i(x)-Y_i(x)\lambda_0(x)dx) \right] \\
&= \E_0\left[\Delta\phi(T) - \int_0^{T}\phi(x)\lambda_0(x)dx\right]=0,
\end{align*}
as stated above.
The martingale property is useful to derive the martingale Stein operator in Eq.~\eqref{eqn:defiT0M}. For more details, see \cite{fernandez2020kernelized}.

\section{Proofs and Derivations}
\label{app:proof}
Proof of Proposition \ref{prop:gksd_quadratic_form}
\begin{proof}
Standard reproducing properties and taking the supremum over unit ball RKHS apply,
\begin{align*}
      \operatorname{Sf-KSD}_{\g}(p \|q; \H) = \sup_{\|\f\|_{\H} \leq 1} \E_p[\left\langle\T_{q,\g}K(x,\cdot), \f \right\rangle_{\H}] = \|\E_p[\T_{q,\g}K(x,\cdot)]\|_{\H}.
\end{align*}
Specifically, assume $f_i(x) = \langle k(x,\cdot) , f_i \rangle$, the setting in \cite{chwialkowski2016kernel,liu2016kernelized}, 
$$\T_{q,\g}K(x,\cdot) = \sum_{i=1}^{d} g_i(x) \left( \frac{\partial \log q(x) }{\partial x^i}k(x,\cdot)  + \frac{\partial k(x,\cdot) }{\partial x^i}\right) + \frac{\partial g_i(x) }{\partial x^i}k(x,\cdot).$$
We can write $h_{q,\g} (x,\tilde x) $ explicitly as
\begin{align*}
h_{q,\g} &(x,\tilde x) 
  =\\
  \sum_{i=1}^d  
  &\left( \frac{\partial^2 k(x,\tilde x)}{\partial x^i\partial \tilde x^i} 
  + \frac{\partial \log q(x) }{\partial x^i}\frac{\partial k(x,\tilde x)}{\partial \tilde x^i} 
  + \frac{\partial \log q(\tilde x) }{\partial \tilde x^i}\frac{\partial k(x,\tilde x)}{\partial x^i} +\frac{\partial \log q(x) }{\partial x^i}\frac{\partial \log q(\tilde x)}{\partial \tilde x^i}  k(x,\tilde x)   \right)\times \\
  & g_i(x)g_i(\tilde x)
  + \frac{\partial g_i(x) }{\partial x^i}\frac{\partial g_i(\tilde x) }{\partial \tilde x^i}k(x,\tilde x).  
\end{align*}
which recovers the quadratic form, which only depends on density $q$ but not $p$.
\end{proof}

{Proof of Theorem \ref{thm:cksd_equivalent}}
\begin{proof}
Note that the expectation on l.h.s. of Eq.~\eqref{eq:survival_stein_equation} is integrating over the density of survival time $\mu_0$ where the expectation on the r.h.s., having the multiplication of $\delta$ in $\T_0 \omega$ in Eq.~\eqref{eqn:survival_stein_operator},  is taken over the paired observation incorporating censoring information. Using the identity in Eq.~\eqref{eqn:expectedPhi0}, we have 
\begin{align*}
    \E_{\mu_0}[\T_{\mu_0,g}\omega] &= 
    \int_{\R_+}  \T_{\mu_0,g}\omega(s) \mu_0(s)ds \\
    &=\int_{\R_+}  \left(\omega'(s) + \omega(s)(\frac{g'(s)}{g(s)}+(\log \mu_0(s))')\right) g(s) \mu_0(s)ds     
    \\
    &=\int_{\R_+} \left(\omega'(s) + \frac{\omega(s)S_C'(s)}{S_C(s)}+\frac{\omega(s)\mu_0(x)'}{\mu_0(s)}\right) \mu_0(s)S_C(s)ds
    \\
    &=\int_{\R_+} \frac{\omega'(s)S_C(s)\mu_0(s) +\omega(s)S_C'(s)\mu_0(s) +\omega(s)S_C(s)\mu_0(x)' }{S_C(s)\mu_0(s)}\mu_0(s)S_C(s)ds\\
    &=\int_{\R_+}  (\T_{0}\omega)(x, \delta)  \mu_0(x)S_C(x)ds
    =\E_{0}[(\T_0 \omega)(x,\delta)] = 0.
\end{align*}

We also note that $S_C(0)=1, S_C(\infty)=0$ by definition of survival functions. As such, $g(x)$ is bounded almost everywhere in $\R_+$ which satisfy the conditions for testing.
\end{proof}

Proof of Theorem \ref{thm:mksd_equivalent}
\begin{proof}
To show the equivalence relation in the sense of Eq.~\eqref{eq:martingale_stein_equation}, we need to consider the presence of indicator variable $\delta$. This is essentially different from the proof of Theorem \ref{thm:cksd_equivalent}. Recall the following identity between hazard function and density:
$\log \mu_0(s) = \frac{\mu_0'(x)}{\mu_0(x)} = \frac{\lambda_0'(x)}{\lambda_0(x)}-\lambda_0(x)$ since
\begin{equation}\label{eq:lamb-to-f}
\frac{\lambda_0'(x)}{\lambda_0(x)}
= \frac{\mu_0'(x)}{S_0(x) \lambda_0(x)} + \frac{\mu_0(x)^2}{S_0(x)^2\lambda_0(x)} =  \frac{\mu_0'(x)}{\mu_0(x)} + \lambda_0(x).
\end{equation}
Denote $\zeta(x) = \frac{\int_0^x \mu_0(s)\omega(s)ds}{\mu_0(x)\omega(x)}$, such that we can write $g = \delta \lambda_0^{-1} + (1-\delta)\zeta$. Decompose the Stein operator $\T_{\mu_0,g}$ w.r.t. $\delta$, we have
\begin{align}\label{eq:cen_uncen_decomp}
    \T_{\mu_0,g}\omega = \delta \T_{\mu_0, \lambda_0^{-1}} \omega +  (1-\delta) \T_{\mu_0, \zeta} \omega
\end{align}
as $\T_{\mu_0}$ is also linear operator w.r.t $g$. We now decompose the above two components $\T_{\mu_0, \lambda_0^{-1}}$ and $\T_{\mu_0, \zeta}$ using the form of Eq.~\eqref{eq:general_stein_operator},
\begin{align*}
    \T_{\mu_0, \lambda_0^{-1}}\omega 
    &= \lambda_0^{-1}\left(\omega' + \omega\log \mu_0' \right) + {\lambda_0^{-1}}'\omega  \\
    & = \lambda_0^{-1}\left(\omega' + \omega(\frac{\lambda_0'(x)}{\lambda_0(x)}-\lambda_0(x)) +\lambda_0 {\lambda_0^{-1}}'\omega \right) \\
    & = \lambda_0^{-1}\left(\omega' - \omega\lambda_0 \right)\\
    &= \lambda_0^{-1}\omega' - \omega
\end{align*}
the second line equality follows from Eq.~\eqref{eq:lamb-to-f} while the third line follows from ${\lambda_0^{-1}}' = -\frac{{\lambda_0}'}{{\lambda_0^{2}}}$. The derivation is interesting that it reveals that the uncensored data in the martingale Stein operator is connected to the Langevin-diffusion via the inverse of hazard function, i.e. when $\delta \equiv 1$ (or absence of censoring), $\T_{\mu_0, \lambda_0^{-1}}\omega = \T_0^{(m)}\omega$.

On the other hand, we rewrite the martingale Stein operator in Eq.~\eqref{eqn:defiT0M} as $(\mathcal T_0^{(m)}\omega)(x,\delta) = \delta\frac{\omega'(x)}{\lambda_0(x)}-\omega(x) =
 \delta(\frac{\omega'(x)}{\lambda_0(x)}-\omega(x))- (1-\delta)\omega(x)$. For Sf-KSD to match this operator, we need to find $\zeta$ such that $\T_{\mu_0, \zeta}\omega = \omega$.
\begin{align}\label{eq:censored_part_DE}
    \T_{\mu_0, \zeta}\omega 
    = \omega\left(\zeta' + \zeta\log \mu_0' \right) + {\omega}'\zeta  =\omega\left(\zeta' + \zeta\log \mu_0' +\zeta(\log \omega)' \right) = - \omega.
\end{align}
As $\omega(x)>0$ for $\mu_0(x)>0$ for $x>0$, Eq.~\eqref{eq:censored_part_DE} gives the following autonomous differential equation form
\begin{equation}\label{eq:censored_part_DE_simplify}
    \zeta' = - \zeta(\log \mu_0' +\log \omega') - 1,
\end{equation}
solving which yields
\begin{align*}
\zeta(x)e^{\log \mu_0(x) + \log \omega(x)}
&=\int_0^{x} -e^{\log \mu_0(s) + \log \omega(s)}ds \\
\zeta(x)\mu_0(x)\omega(x) &= \int_0^x \mu_0(s)\omega(s)ds
\end{align*}
as $\omega(0) = 0$ by assumption.
$\zeta(x) = \frac{\int_0^x \mu_0(s)\omega(s)ds}{\mu_0(x)\omega(x)}$ as proposed.
Putting together, we have
$$
\T_{\mu_0,g}\omega = \delta \T_{\mu_0, \lambda_0^{-1}} \omega +  (1-\delta) \T_{\mu_0, \zeta} \omega
=  \delta(\frac{\omega'(x)}{\lambda_0(x)}-\omega(x))- (1-\delta)\omega = \T^{(m)}_0 \omega
$$
and Stein's identity result follows by taking the expectation of the same form, $\E_0[\T_{\mu_0,g}\omega] = \E_0[\T^{(m)}_0 \omega]=0$.
\end{proof}


Proof of Corollary \ref{cor:sksd_equivalent}
\begin{proof}
The proof follows from decomposing the survival Stein operator $\mathcal{T}_0^{(s)}\omega$ in to the uncensored part and censored part, similar to Eq.~\eqref{eq:cen_uncen_decomp},
\begin{align*}
(\mathcal{T}_0^{(s)}\omega)(x,\delta)
&=\delta\omega'(x)+\delta\omega(x)\frac{\lambda_0'(x)}{\lambda_0(x)}-\omega(x)\lambda_0(x)\\
&=\delta\omega'(x)+\delta\omega(x)\left(\frac{\lambda_0'(x)}{\lambda_0(x)}-\lambda_0(x)\right)-(1-\delta)\omega(x)\lambda_0(x)\\
&=\delta\omega'(x)+\delta\omega(x)\frac{\mu_0'(x)}{\mu_0(x)}-(1-\delta)\omega(x)\lambda_0(x)\\
&=\delta\Big( \omega'(x)+ \omega(x) \log \mu_0(x)'\Big) -(1-\delta)\omega(x)\lambda_0(x).
\end{align*}
where the term involving $\delta$, the uncensored part, is just the Langevin-diffusion Stein operator in 1d. Similar to Eq.~\eqref{eq:censored_part_DE}, we solve the following autonomous differential equation for the censored part,
\begin{align}
    \T_{\mu_0, \zeta}\omega 
    =\omega\left(\zeta' + \zeta\log \mu_0' +\zeta(\log \omega)' \right) = - \omega \lambda_0.
\end{align}
which simplifies to 
\begin{equation}
    \zeta' = - \zeta(\log \mu_0' +\log \omega') - \lambda_0.
\end{equation}
Solving the differential equation with boundary condition $\omega(0) = 0$, we get 
 $\zeta(x) = \frac{\int_0^x \mu_0(s)\omega(s)\lambda_0(s)ds}{\mu_0(x)\omega(x)}$. As such, using $g = \delta + (1-\delta)\zeta$, the result follows
$$
\T_{\mu_0,g}\omega = \delta \T_{\mu_0} \omega +  (1-\delta) \T_{\mu_0, \zeta} \omega = 
(\mathcal{T}_0^{(s)}\omega)(x,\delta).
$$
\end{proof}

\paragraph{Remarks}
In the main text, we discussed the advantages and disadvantages of KSD-based test with $\mathcal T^{(m)}_0 $ and $\mathcal T^{(s)}_0 $, which corroborate the empirical findings in \cite{fernandez2020kernelized}.
Moreover, \cite{fernandez2020kernelized} studied the testing procedure via c.d.f. transformation 
followed by testing the uniform null density,
which they call \textbf{model-free implementation}. This procedure has been shown to achieve higher test power.
Similar testing strategy via c.d.f. transformation has been studied in \cite{fernandez2019maximum} using MMD-based test. 
Notice that since $F_0$ is monotone and $u_i=F_0(t_i)=\min\{F_0(x_i),F_0(c_i)\}$, thus $\delta_i$ remains consistent. 
Under this transformation, the null hypothesis is equivalent to test whether $F_0(x_i)$ is distributed as a uniform random variable.  In this setting, the observations for the test is based on $\{(u_i,\delta_i)\}_{i\in[n]}$, where the Stein operator used is independent of density of $x\sim f_0$. 
Instead, $\mu_0(u)\equiv 1$ and
$\lambda_0=\lambda_{\mathcal{U}}=\frac{1}{1-x}$ are used to construct the Stein operator and 
$$(\mathcal{T}_0^{(m)} {\omega})(u,\delta)=\delta {\omega'(u)}(1-u) - \omega(u)$$
for $u=F_0(x)$  (notice that $F_0(0) = 0$). 
From our result, we see that with the particular choice of the uniform null, there is no more interaction between test function and density in the censored part, e.g. $\zeta(x) = \frac{\int_0^x \omega(s)/(1-s)ds}{\omega(x)}$, resulting a better estimation accuracy from the samples, thus higher test power.

Similarly, \cite{fernandez2020kernelized} exploited another monotone transformation via  the cumulative hazard function from the null $\Lambda_0$, such that $\Lambda_0(X) \sim \textnormal{Exp}(1)$. In this case, $\mu_0(x) = \exp(-x)$ which still require interaction between $\mu_0$ and $\lambda_0$ in $\zeta$, resulting in decrease in estimation accuracy. Our results explain the empirical finding in \cite{fernandez2020kernelized} that the test power from the model-free implementation of the test using cumulative hazard transformation is not higher than using the c.d.f. transformation.
\newline
\newline
Proof of Theorem \ref{thm:ssd_equivalent}
\begin{proof}
As $g_i$ consists of finite sum of delta measures on locations $z_j$ sampled from $q(z|x)$, it's derivative is $0$ everywhere excluding a finite number of points which is a set of measure zero. 
Recall that the marginalisation of the density $q(x) = \int q(x,z)dz$,
\begin{align*}
  \E_q[\T_{q, \g} \f] 
  &= \int \sum_i \T_q(g^z_i(x)f_i(x)) q(x,z) dx dz  \\
  &\overset{(a)}{=} \int_{\Z} \sum_i (\int_{\X} \left(\log q(x) '(g^z_i(x)f_i(x)) +(g^z_i(x)f_i(x))' \right)q(x|z) dx) \pi(z) dz  \\
 & \overset{(b)}{=} \sum_i \int \sum_j 
 \left(\log q(x|z_j) 'f_i(x) +f_i(x)' \right) q(x|z_j) dx \\
 & = \sum_j \E_{q(x|z_j)} [\T_{q(x|z_j) } \f] 
 = \E_{q}[\T_{q,z}\f]
\end{align*}
where equality (a) uses the marginalisation and equality (b) utilises the the fact that $z_j$ are samples from $q(z|x)\propto \pi(z)q(x|z)$ and the bivariate delta function gives $\int \delta_{z_j}(x,z) q(x,z) dz  = q(x,z_j)=q(x|z_j)\pi(z_j)$.  
\end{proof}

Proof of Theorem \ref{thm:second_order_equivalent}
\begin{proof}
It is not difficult to see that
\begin{align}\label{eq:second_order_derivation}
\T_{q} ( f g )=\T_{q} ( f (\log f)' ) =\T_{q} ( f' )    = f''
+ \log q' f'    
\end{align}
which is the second-order operator in 1d. 

For the multivariate case,
choosing $\g:\R^d \to \R^d$ such that
$g_i(x)=\frac{\partial}{\partial x^i}f_i(x)\partial x^i$, where $\partial x^i$ is the differential form defined in the main text. Similar derivation as Eq.~\eqref{eq:second_order_derivation} above, we have
$$
\T_{q,\g} \f = \sum_{i} \T_{q,(\frac{\partial}{\partial x^{i}} \log_i f_i)} f_i  = \sum_{i} \frac{\partial^2}{\partial {x^{i}}^2} f_i \partial x^i + \log q' \frac{\partial}{\partial x^{i}}f_i \partial x^i =  \T_{q} \nabla \cdot \f.
$$
\end{proof}

\begin{Theorem}[Characterisation of $\operatorname{bd-KSD}$]\label{thm:characteristic}
Let $\tilde p$,  $\tilde q$ be smooth densities defined on $V$. 
Assume:  1) kernel $k$ is compact universal 
\textnormal{\cite[Definition 2(ii)]{carmeli2010vector}}; 2)
$\E_{x,x' \sim \tilde q} [h_{q, \g}(x,{x}')^2] < \infty$; 
3) $\E_{\tilde q} \| L(x) \|^2 < \infty$; 4) $g_i(x)>0$ whenever $q(x)>0$. 
Then, $\operatorname{bd-KSD}_{\g}(\tilde q\|\tilde p )\geq 0$ and $\operatorname{bd-KSD}_{\g}(\tilde q\|\tilde p )=0$ if and only if $\tilde p=\tilde q$.
\end{Theorem}

Proof of Theorem \ref{thm:characteristic}

\begin{proof}
The Theorem extends from \cite[Theorem 2.2]{chwialkowski2016kernel} with additional assumptions $g_i(x)>0$ if $\tilde q(x)>0$, together with appropriate compact universality condition for the kernel. For more general settings, having a proper notion of universal kernel would extend Theorem \ref{thm:characteristic} to show that bd-KSD is a proper discrepancy in desired testing scenarios.

Denote $\textbf{s}_{q,\g}(\cdot) = \E_{{x}\sim p}[\widetilde{\mathcal T}_{q,\g}  k({x},\cdot)]\in \mathcal{H}$ and we can write the quadratic form bd-KSD as $\operatorname{bd-KSD}_{\g}(\tilde q\|\tilde p )^2=\|\textbf{s}_{q,\g}(\cdot)\|^2_{\H} \geq 0$. 
If $p=q$, then $\operatorname{bd-KSD}_{\g}(\tilde q\|\tilde p )^2=0$ from Stein's identity. 

Conversely, if $\operatorname{bd-KSD}_{\g}(\tilde q\|\tilde p )^2=0$, then $\textbf{s}_{q,\g}(x)=\textbf{0}$, $\forall x$, s.t. $p(x)>0$.
Then, from $\log (q/p) = \log (\tilde q) - \log (\tilde p)$, we obtain,
\begin{align*}
    &{\E}_{{x}' \sim \tilde p} \left[ L_i({x}') k({x}',x) \right] = (\textbf{s}_{q,\g})_i(x) - {\E}_{{x}' \sim \tilde p} \left[ \widetilde{\mathcal T}_{q,\g} k({x}',x)  \right] = 0,
\end{align*}
for every $x$ with positive densities. As $g_i(x)>0$ for $q(x)>0$, and $k$ is compact-universal at $V$, the injectivity result in \cite[Theorem 4(b)]{carmeli2010vector} implies that $L_i=0, \forall i\in[d]$.
Therefore, $\log (\tilde q/\tilde p)$ is constant on $V$.
Since both $\tilde p$ and $\tilde q$ are both densities on $V$ that integrate to one, we conclude $\tilde p=\tilde q$.
\end{proof}

\section{Additional Simulation Results}\label{app:exp}

We investigate the test performances
on various problems, comparing different choice of $\g$. In the following problems, we apply the Gaussian kernel with median distance as bandwidth \citep{gretton2012kernel}.
\paragraph{1. Truncated Gaussian distribution in $B_1(\R^3)$}
\begin{equation}\label{eq:truncated_gauss}
\tilde q_{\nu}(x)\propto \mathcal{N}(x|0,\Sigma_{\nu}), \forall x\in B_1(\R^3)
, \Sigma_{\nu} = \begin{pmatrix}
1 & \nu &0 \\
\nu & 1 &0 \\
0 & 0 &1 
\end{pmatrix}    
\end{equation}
where $\nu > -1$.
We test the null model $\tilde{q}_0$
against the alternative $\tilde{q}_1$ with perturbation of variance parameter $\nu = 1.0$. The test power of bd-KSD, with different choice of  $g^{(p)}$, is shown in Fig.~\ref{fig:truncation_normal}.
All tests have increasing test powers as the simple size increases, which is what we expect.
As the alternative is having variance difference in the first direction $x^1$, the test captures such difference better when more weights are put on near the origin (refer Fig.~\ref{fig:truncation_aux}
for $g(x)$ values).   
Hence, the test power increases with the increase of parameter $(p)$ for auxiliary function $g^{(p)}$. 
For $p>1$, the bd-KSD tests outperforms the MMD test\footnote{For all MMD-based tests, we draw $n$ samples from the null samples when the sample size of observed data is $n$.}. 

\begin{table}[ht!]
    \centering
\begin{tabular}{lrrrr}
\toprule 
{} &   n=100 &     n=400 &   n=700 &     n=1000 \\
\midrule
bd-KSD(p=4) &  0.004 &  0.010 &  0.030 &  0.018 \\
bd-KSD(p=2) &  0.008 &  0.018 &  0.022 &  0.012 \\
bd-KSD(p=1) &  0.010 &  0.016 &  0.020 &  0.028 \\
bd-KSD(p=0.5) &  0.006 &  0.008 &  0.010 &  0.030 \\
\midrule
MMD &  0.022 &  0.014 &  0.018 &  0.032 \\
KSD($g(x) \equiv1 $) &  1.00 &  1.00 &  1.00 &  1.00 \\
\bottomrule
\end{tabular}
    \caption{Rejection rate under the null; $\alpha = 0.01$; 500 trials. }
    \label{tab:truncation_typeI}
\vspace{-0.1cm}
\end{table}

The rejection rate under the null is reported in Table~\ref{tab:truncation_typeI}. The bd-KSD tests, with appropriately choice of $g$ incorporating the boundary conditions, achieve well-controlled Type-I errors. MMD-based tests also achieves the well-controlled Type-I errors.
However, applying the KSD tests in Eq.~\eqref{eq:KSDequiv} are unable to have controlled test level due to the violation of Stein's identity, which is what we expect. 

\paragraph{2. Truncated Gaussian mixture in $B_1(\R^2)$}
$$
\tilde{q}_{\nu} (x) \propto \frac{1}{2}\mathcal{N}(x|\mu_1,\Sigma_{\nu}) + \frac{1}{2}\mathcal{N}(x|\mu_2,\Sigma_{\nu}),
\forall x\in B_1(\R^2),
$$
where $\mu_1=(-1,0,0)$, $\mu_2=(1,0,0)$ and $\Sigma_{\nu}$ (as defined above) is shared between two components.
We test the null model $\tilde{q}_0$
against the alternative $\tilde{q}_{-\frac{1}{2}}$.
The test power with increasing sample size is shown in Fig.~\ref{fig:truncation_gmm}. Similar as the previous case, the bd-KSD of $g^{(p)}$ with larger parameter value $p$ achieves better test power. MMD-based tests perform slightly better that bd-KSD with $g^{(1)}$. However, MMD suffers from slow computational time due to the sampling procedure in a bounded domain, as shown in Fig.~\ref{fig:truncation_runtime}.

\begin{figure*}[t!]
    \centering
    \subfigure[Gaussian distributions]{
    \includegraphics[width=0.235\textwidth]{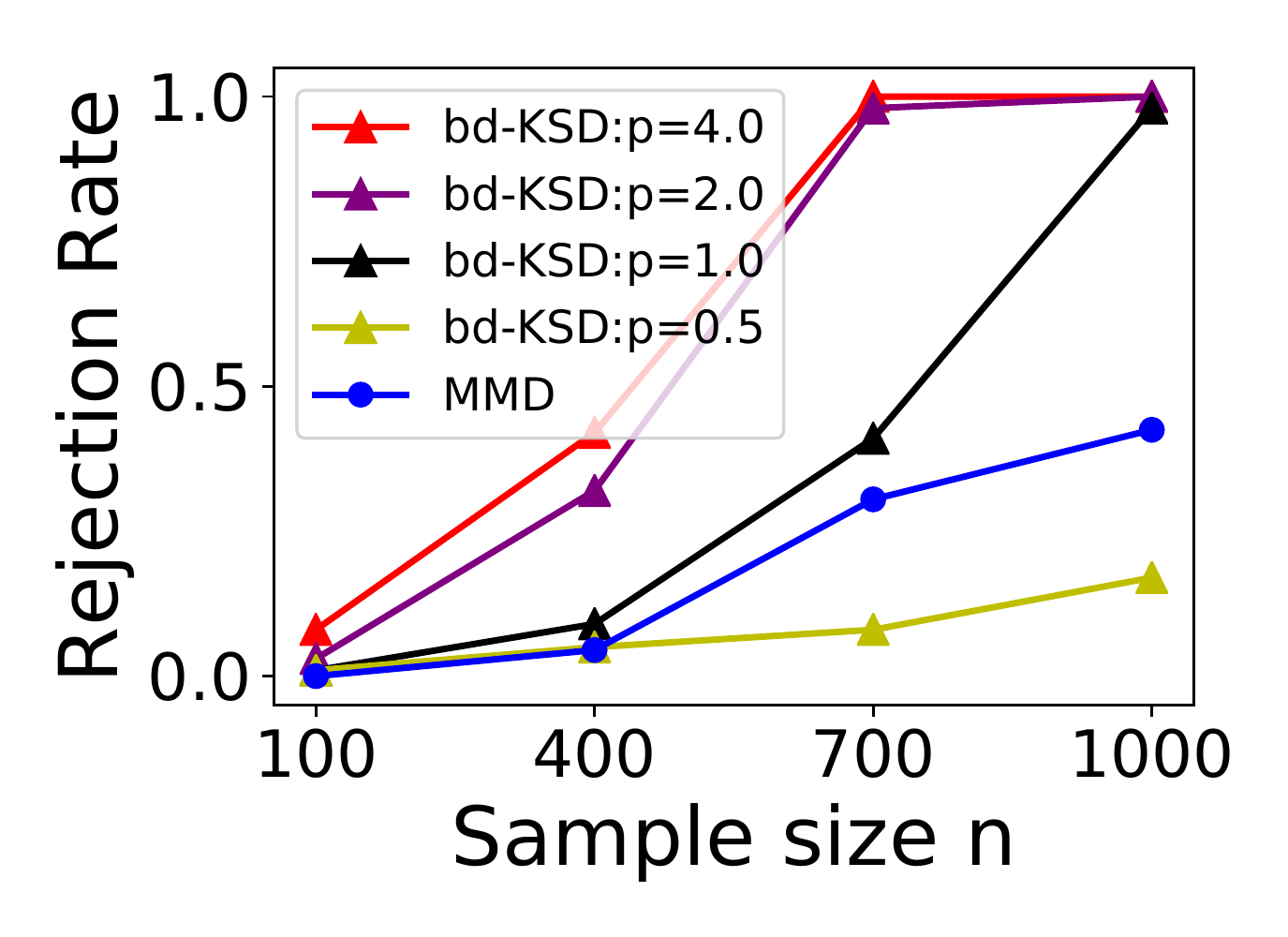}\label{fig:truncation_normal}}
    \subfigure[Gaussian mixture 
    ]{
    \includegraphics[width=0.235\textwidth]{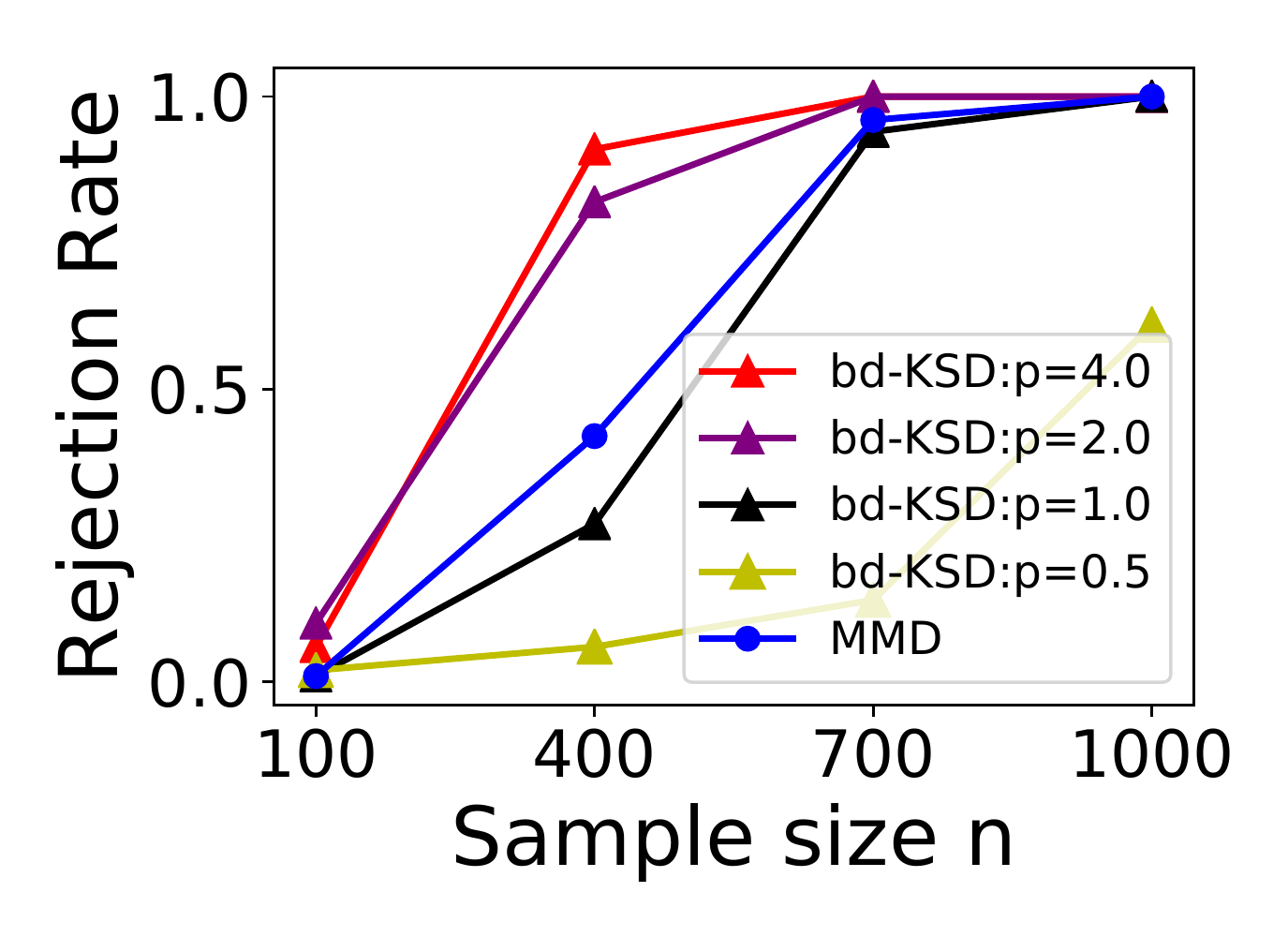}\label{fig:truncation_gmm}}
        \subfigure[Computation runtime]{
    \includegraphics[width=0.235\textwidth]{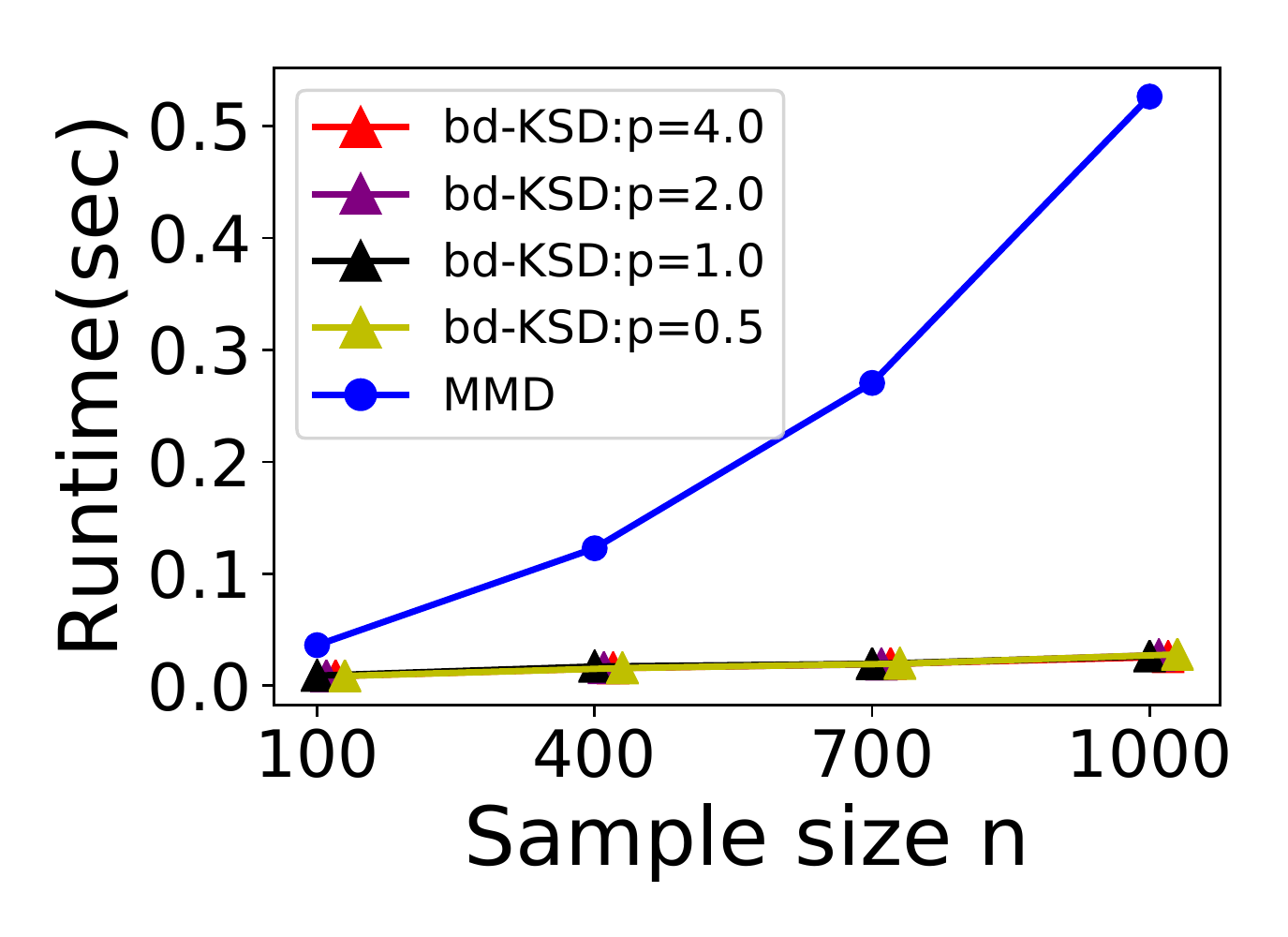}\label{fig:truncation_runtime}}
    \subfigure[Auxiliary function]{
        \includegraphics[width=0.235\textwidth]{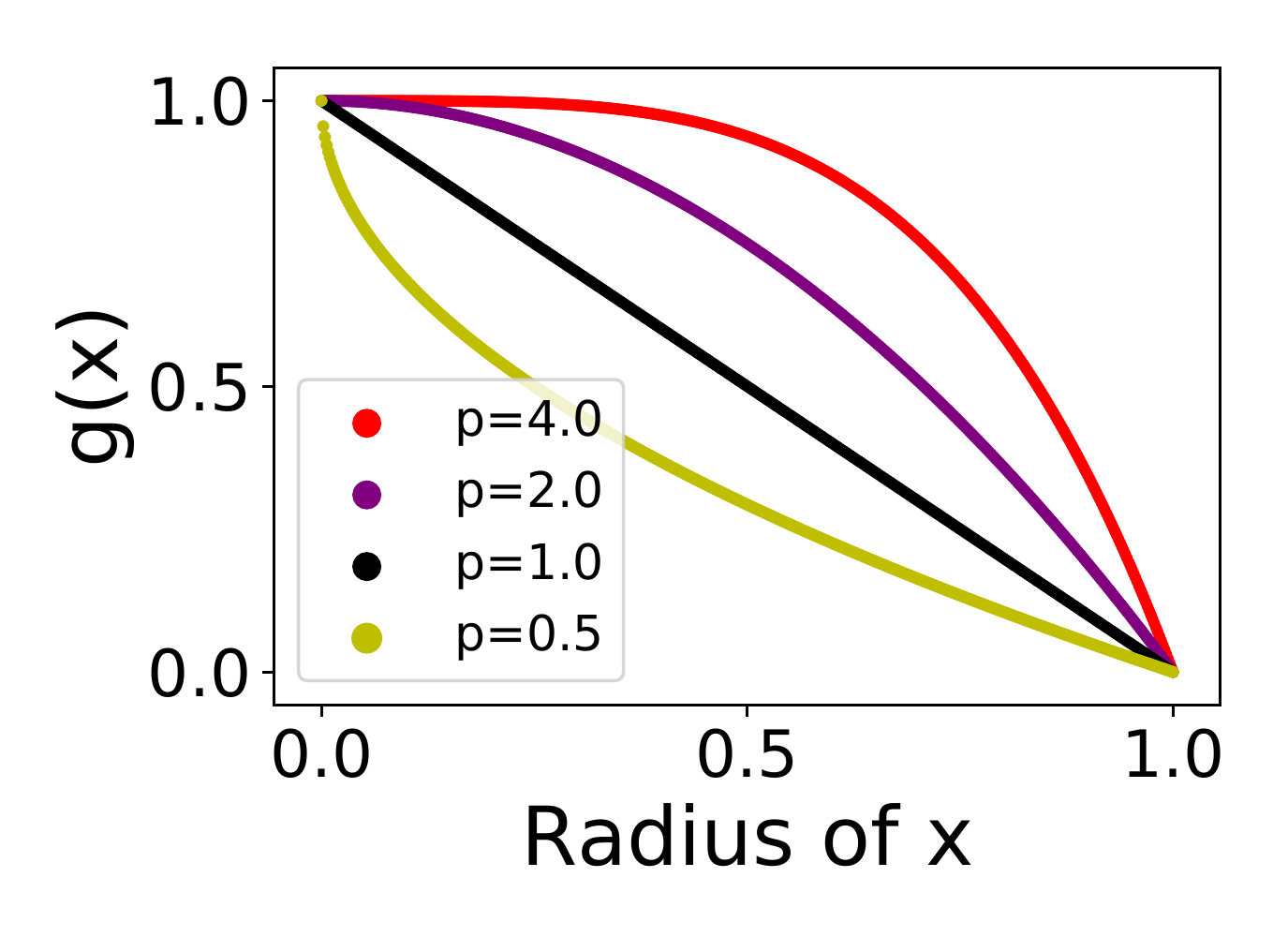}\label{fig:truncation_aux}}
    \caption{Truncated distribution in unit ball $B_1(\R^3)$; test level $\alpha = 0.01$; test repeats $200$ trials.}
    \label{fig:truncation}
\end{figure*}

\begin{figure}[t!]
    \centering
    \subfigure
    {
    \includegraphics[width=0.253\textwidth]{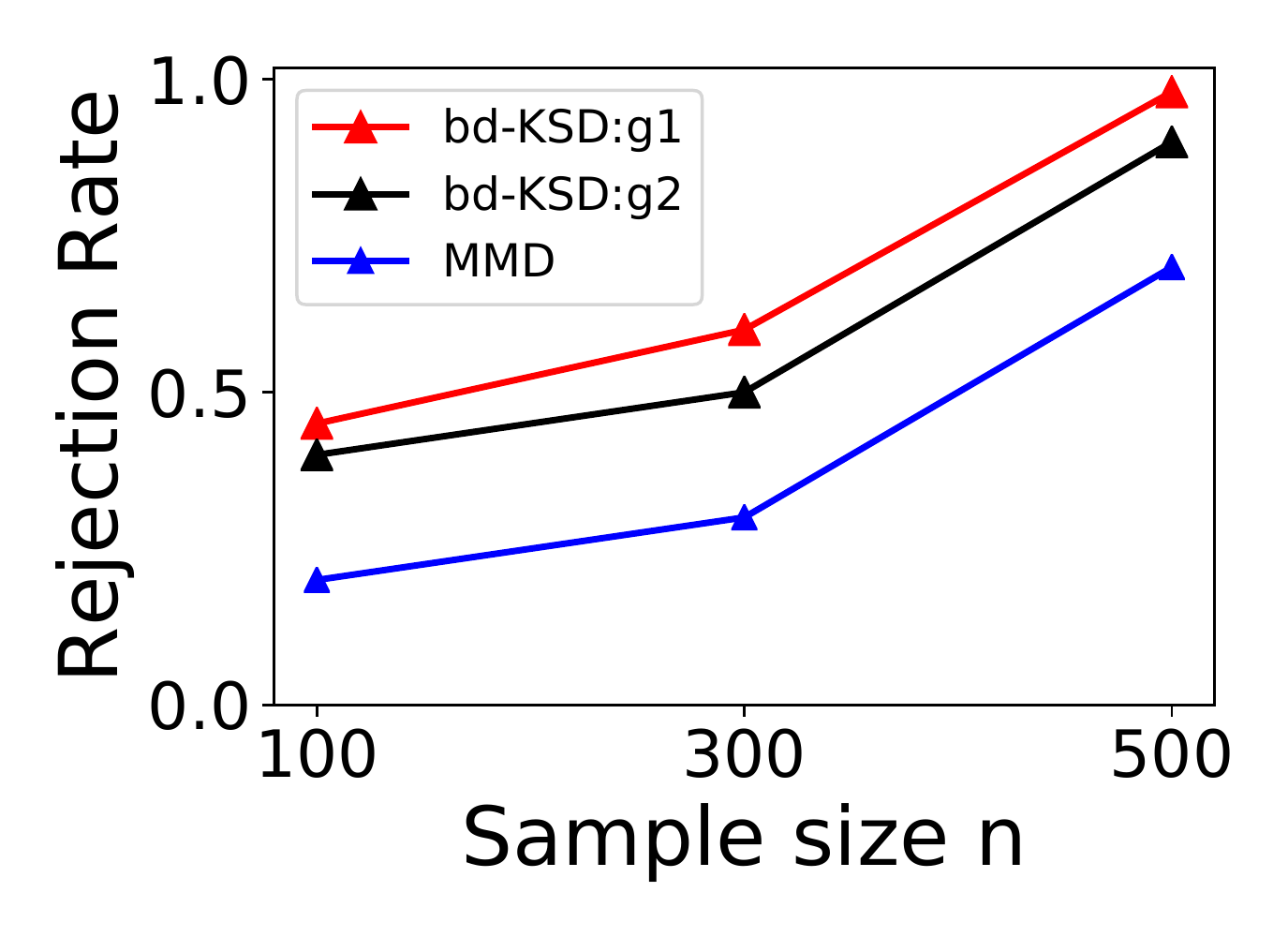}\label{fig:composition_asy}}
    \subfigure
    {
    \includegraphics[width=0.253\textwidth]{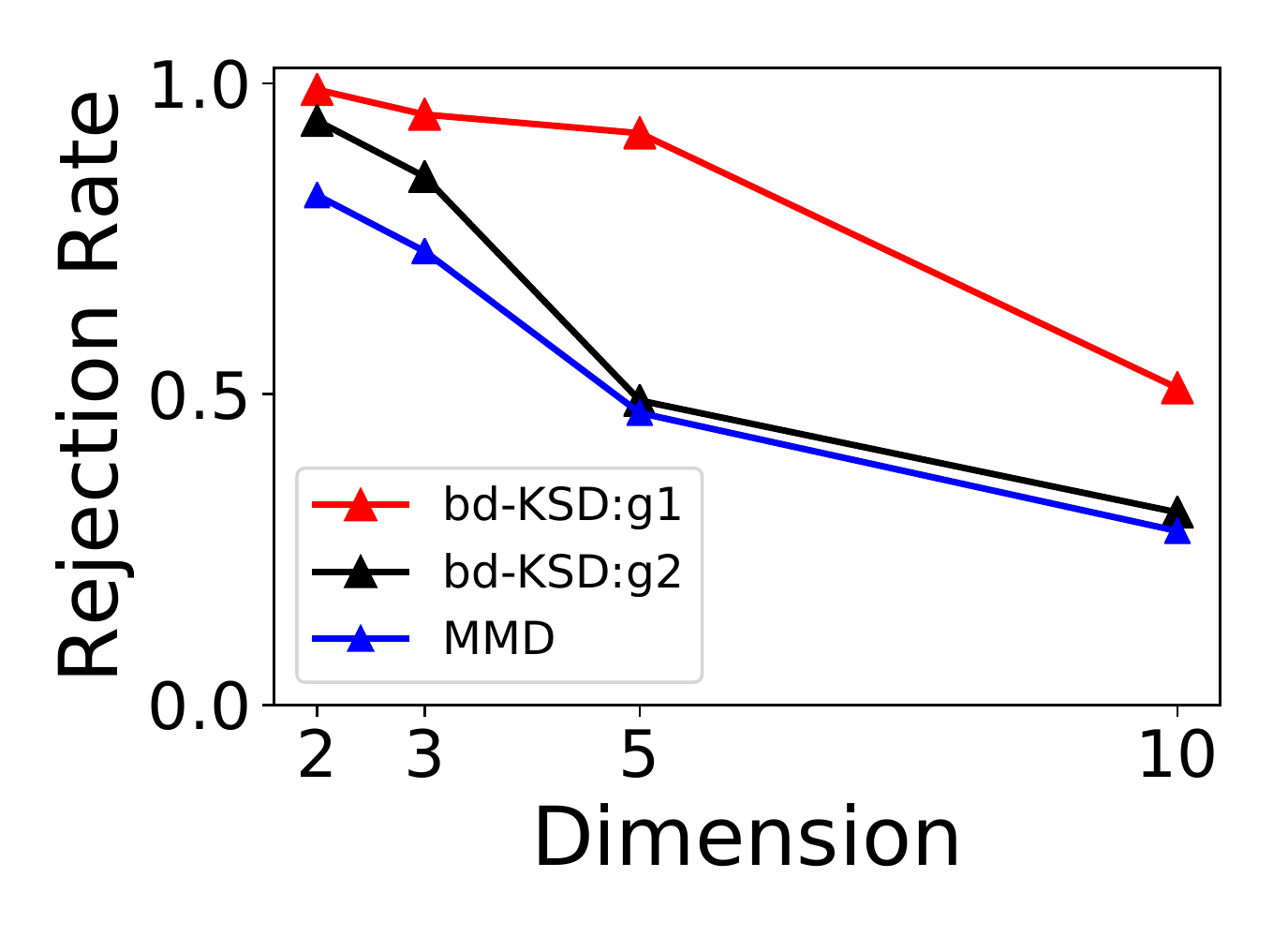}\label{fig:composition_d}}
    \caption{ Dirichlet distributions on simplex ${\rm S}^{d-1}$; test level $\alpha = 0.01$; test repeats 200 trials.}
    \label{fig:compositional}
\end{figure}

\paragraph{3. Dirichlet distributions in ${\rm S}^{d-1}$}
\begin{equation}\label{eq:dirichlet}
    \tilde q_{\nu}(x)\propto (x^1)^{\nu}\cdot\Pi_{i=1}^d (x^i)^{(-\frac{1}{2})}, \forall x\in {\rm{S}}^{d-1}.
\end{equation}
We test the null model $\tilde{q}_0$
against the alternative $\tilde{q}_{-\frac{1}{3}}$ with perturbation of the concentration parameter in the first dimension $(x^1)$. The test power of bd-KSD, with the choice of geometric mean function $g^{(1)}$ and the minimum distance-to-boundary function $g^{(2)}$, and the MMD-based test are shown in Fig.~\ref{fig:compositional}. From the result, we see that the bd-KSD tests have higher test power compared to MMD-based test power. The test power increases as sample size increases (Fig.~\ref{fig:compositional} left) and decreases as the dimension of the problem increases (Fig.~\ref{fig:compositional} right), which is what would we expect. We also see that 
geometric mean function 
$g^{(1)}$ induces a better Stein operator for bd-KSD testing on compositional data compared to 
the minimum distance-to-boundary function 
$g^{(2)}$.

\section{Additional KSD Details}\label{supp:more_KSD}

\subsection{{ Kernel Discrete Stein Discrepancy}}\label{supp:KDSD_supp}

In this section, {we}  briefly review the kernel discrete Stein discrepancy {(KDSD)} {introduced in \citep{yang2018goodness}. First we need some definitions.} 
\begin{Definition}\label{def:cyclic_perm}[Definition 1 \citep{yang2018goodness}](Cyclic permutation). For a set X of finite cardinality,
a cyclic permutation $\neg: X \to X$ is a bijective function
such that for some ordering $x^{[1]}, x^{[2]}, \dots , x^{[|X |]}$ of the
elements in $X$ , $\neg x^{[i]} = x^{[(i+1) mod |X |]}, \forall i = 1, 2, \dots , |X |$.
\end{Definition}

\begin{Definition}\label{def:difference_operator}[Definition 2 \citep{yang2018goodness}]
Given a cyclic permutation $\neg$ on $X$ , for any d-dimensional vector $x = (x_1, . . . , x_d)^{\top} \in X^d$
, write $\neg_i x := (x_1, \dots , x_{i-1}, \neg x^i, x_{i+1}, \dots, x_d)^{\top}.$ 
For any function $f : X^d \to \mathbb{R}$, denote the (partial) difference operator as
\begin{align*}
    \Delta_{x^i}f(x): = f(x) - f(\neg_i x), \quad \quad  i=1,\dots,d
\end{align*}
and {introduce} the difference operator:
\begin{align*}
    \Delta_\neg f(x): = ( \Delta_{x_1} f(x), \dots ,  \Delta_{x_d} f(x))^{\top}.
\end{align*}
\end{Definition} 
{Here we use the notation $ \Delta_\neg$ to distinguish it from the notation in the main text, where we used  
$\Delta_s h(x) = h(x^{(s,1)}) - h(x^{(s,0)})$ and $ || \Delta h ||  = \sup_{s \in [N]} | \Delta_s h(x)|.$
}

For discrete {distributions $q$}, \citep{yang2018goodness} propose {the following} discrete  Stein operator, which is based on the difference operator {$  \Delta_\neg $} constructed from  a cyclic permutation: 
\begin{align}
\A^{D}_q f(x) = f(x)\frac{\Delta_\neg q(x)}{q(x)} - \Delta_\neg^{\ast} f(x), \label{eq:disrete_stein}
\end{align}
where $\Delta_\neg^{\ast}$ denotes the adjoint operator of $\Delta_\neg$.

In \cite{yang2018goodness}, the generalisation that better characterises the density $q$ is stated in the following form,
\begin{align}
\A^{D}_{q,\L} f(x) = f(x)\frac{\L q(x)}{q(x)} - \L^{\ast} f(x), \label{eq:disrete_stein_L}
\end{align}
where $\L^{\ast}$ is the adjoint operator of $\L$.

\subsection{Stochastic Stein Discrepancy}\label{app:ssd}
\cite{gorham2020stochastic} proposed stochastic Stein discrepancy (SSD) via the following subset operators. 

Given prior $\pi_0$, likelihood $\pi(\cdot|x)$ and samples
$y_1,\dots,y_L$, the posterior density $q(x) \propto \pi_0(x) \Pi_{l=1}^L \pi(y_l|x)$. 
With uniformly sampled index set $\sigma \subset [L]$ with $|\sigma|=m$, the stochastic Stein operator is defined as
\begin{equation}\label{eq:ssd_operator}
    \T_{\sigma} \f(x) = \frac{L}{m} \f (x)^{\top} \nabla \log q_{\sigma}(x)+ \nabla \cdot \f(x)
\end{equation}
for test function $\f$ and $q_\sigma(x):=\pi_0(x)^{m/L} \Pi_{l\in \sigma} \pi(y_l|x)$.
The stochastic Stein variational gradient descent (SSVGD) is then developed for sampling procedures and nice convergence properties has been shown in \cite{gorham2020stochastic}.

where the latent samples $z_j$ are functionally analogous to the observations $y_l$ in SSD. The key difference is that for every single $x$ in SSD, multiple $y_l$ acts on it; while for latent-variable Stein operator, only one $z_j$ acts on it at a time.


\section{Additional Generalisation via Standardisation-functions}\label{app:linear_generalisation}
The choice of the Langevin-diffusion type of Stein operator in Section \ref{sec:stein_operators} is not unique and many other Stein operators can characterise the same distribution. 
A particular method to generalise the Stein operator in Eq.~\eqref{eq:steinRd}, is via an appropriate linear operator $\L$ acting on the test function $f$, i.e. 
\begin{equation}\label{eq:general_stein_with_linear_operator}
    \T_{q, \L} f = \T_q (\L f) = \L f^{\top} \nabla \log q + \nabla \cdot \L f.
\end{equation} 
Some related ideas involving $\L$ to generalise learning objectives for unnormalised model have been discussed in the context of score matching \citep{lyu2012interpretation}\footnote{In \cite{lyu2012interpretation}, $\L$ acts on density $q$ instead of the test function $f$ here, where a common choice of $\L$ is \emph{marginalization} operator.}. \cite{yang2018goodness}\footnote{In \cite{yang2018goodness}, $\L f(x) = \sum_{x'} l(x,x')f(x')$, for discrete variable $x$ and bivariate function $l$. Different from \cite{lyu2012interpretation}, $\L$ may act on both the probability mass function or the test function. For the particular form of discrete KSD studied in \cite{yang2018goodness} $\L$ is chosen as the partial difference operator, which is essentially different from the diffusion based generalisation in Eq.~\eqref{eq:general_stein_with_linear_operator} for continuous variables. 
More details 
are included in Appendix~\ref{supp:KDSD_supp}.}
suggests similar formulation for characterising discrete KSD, while not yet investigated. 

In the presented work, we focus on a class of new Stein operators derived from Eq.~\eqref{eq:general_stein_with_linear_operator} where the linear operator $\L$ is chosen as the element-wise product using a vector-valued function $\g$ that we call the \emph{auxiliary function}. Even though this is a subclass of Stein operators in Eq.~\eqref{eq:general_stein_with_linear_operator},  we show that, with specific choice of $\g$, the class of Stein operators in this particular form is just enough to generalise the set of  Stein operators introduced in Section \ref{sec:stein_operators}.

Such formulation can be more general than the element-wise product cases developed in the main text. However, the elementwise product formulation is the simplest case to generalise existing Stein operators for goodness-of-fit test.

To see the interplay between Eq.~\eqref{eq:general_stein_with_linear_operator} and Eq.~\eqref{eq:general_stein_operator}. We use the generalisation of second order operator in Theorem \ref{thm:second_order_equivalent} as an example. 
Multivariate notion utilises the $\nabla$ notation as defined in the main text. In the $\R^d$ case, the metric tensor terms $[G^{-1}]_{ij}=\delta_{i=j}$ so that the 
We show here the more interesting scenario for the Riemannian manifold case, where
the choice of generalised Stein operator corresponds to the second-order operator
incorporating the Riemannian metric. 

As $[G^{-1}]_{ij}$ may not vanish when $i\neq j$ in the Riemannian manifold scenario, it is not possible to generalise the second-order Stein operator for Riemannian manifold with the form of Eq.~\eqref{eq:general_stein_operator} using elementwise product between vector-valued functions; however, with the more general formulation in Eq.~\eqref{eq:general_stein_with_linear_operator}, we are able to show this.

\begin{Corollary}\label{cor:second_order_equivalent_riemannian}
For scalar test function $\tilde f:\M \to \R$,
choosing $\L{\tilde f} (x)= \sum_{i,j} [G^{-1}]_{i,j} \frac{\partial}{\partial x^j} \tilde f(x) \partial x^i $, the Stein operator in the form of Eq.~\eqref{eq:general_stein_with_linear_operator} recovers the second-order differential operator defined in Eq.~\eqref{eq:second_order_operator} for Riemannian manifold,
$\T_{q,\L} \tilde f = \T^{(2)}_q \tilde f$.
\end{Corollary}
\begin{proof}
By definition, we know that $\L \tilde  f = \nabla \tilde f$ incorporating the Riemannian metric $G$. Thus,  
$$\T_{q,\L} \tilde f =\T_{q} \nabla \tilde f = \nabla \tilde f ^{\top} \nabla \log q(x)+\nabla \cdot \nabla \tilde f =  \T^{(2)}_q \tilde f$$
by construction. And the $\nabla$ and $\Delta$ notation is also w.r.t. the Riemannian manifold $\M$ given metric tensor $G$.
\end{proof}

We note from the Corollary \ref{cor:second_order_equivalent_riemannian} that, for each vector direction $\partial x^i$, it is summed over all possible differential operator, $\frac{\partial}{\partial x^j}$, acting on \emph{scalar} test function $\tilde f$, instead of just using $\frac{\partial}{\partial x^i}$. Hence, the element-wise operation over vector-valued test function $\f$ in Eq.~\eqref{eq:general_stein_operator} does not generalise this form. An additional note to combine the element-wise product form of Eq.~\eqref{eq:general_stein_operator} and the second-order Stein operator on Riemmanian manifold can be the following.


\section{Additional Interpretations and Comparisons}\label{supp:comparison}

\subsection{Calculus of Variation for Optimality Conditions}
We review the calculus of variation and apply techniques using Euler-Lagrangian equation that we used for optimality condition in Eq.~\eqref{eq:euler-lagrangian}. Consider some loss functional in the integral form
$$
J[y] = \int_{\X} L(y(x),\dot{y}(x))  
dx.
$$
In one dimension, $\X$ can be finite interval as well. 
The idea for calculus of variation is to perturb the test function $y$ for a small amount and find the stationary point for the loss functional. Let $\eta(x)$ be any function vanishing at the boundary. We consider, $y+\varepsilon \eta$, 
and letting $\varepsilon \to 0$. Taking derivative w.r.t. $\varepsilon$ we get
$$
\frac{dL}{d\varepsilon} = \frac{\partial L}{\partial y}\frac{d y}{d\varepsilon} +  \frac{\partial L}{\partial y'}\frac{d y'}{d\varepsilon} =  \frac{\partial L}{\partial y}\eta +  \frac{\partial L}{\partial y'}\eta' .
$$
Then we consider $\int_{\X}\frac{dL}{d\varepsilon} |_{\varepsilon=0}$ and using integration by parts, we have
\begin{equation}\label{eq:cov_der}
\frac{d}{d\varepsilon} |_{\varepsilon=0}\int_{\X} L dx = \int_{\X} \left( \frac{\partial L}{\partial y}\eta +  \frac{\partial L}{\partial y'}\eta'  \right) dx = \int_{\X}  \frac{\partial L}{\partial y}\eta + \frac{\partial L}{\partial y'}\eta'|_{\partial X}  - \int_X \eta \frac{d}{dx}\frac{\partial L}{\partial y'} dx 
= \int_{\X}  \eta\left(\frac{\partial L}{\partial y} - \frac{d}{dx}\frac{\partial L}{\partial y'}\right) dx.
\end{equation}
Setting the Eq.~\eqref{eq:cov_der} above to $0$ we get $\int_{\X}  \eta\left(\frac{\partial L}{\partial y} - \frac{d}{dx}\frac{\partial L}{\partial y'}\right) dx = 0$. As $\eta$ is chosen to be any perturbation function, we then have 
$$
    \frac{\partial}{\partial y}L(y(x),\dot{y}(x)) - \frac{d}{dx} \frac{\partial}{\partial \dot{y}}L(y(x),\dot{y}(x)) = 0,
$$
which is referred to as the Euler-Lagrangian equation. 
And we will see that such techniques apply to our multivariate case when we perturb our test function in each direction.

Proof of Theorem \ref{thm:optimal}
\begin{proof}

Let $L = \E_p[\T_{q,\g}\f]$ and apply the Euler-Lagrangian equation in Eq.~\eqref{eq:euler-lagrangian} perturbing $f_i$ w.r.t. $x^i$, i.e. with $y$ in Eq.~\eqref{eq:euler-lagrangian} as $f_i$, we have condition
$$
    \frac{\partial L}{\partial f_i} - \frac{d}{dx^i} \frac{\partial L}{\partial \dot{f_i}} = 0,
$$
for all $i \in [d]$. 
Referring back to the form of Eq.~\eqref{eq:general_stein_operator}, 
$$
(\T_{q, \g} \f) 
=\sum_{i=1}^d g_i (f_i \frac{\partial}{\partial x^i} \log q 
+ \frac{\partial}{\partial x^i} f_i)
+ f_i \frac{\partial}{\partial x^i}g_i,
$$
we get 
$\frac{\partial L}{\partial f_i} = \E_p[g_i \frac{\partial}{\partial x^i} \log q +     \frac{\partial g_i}{\partial x^i}]$ and  $\frac{\partial L}{\partial \dot{f_i}} = \E_p[g_i]$. So the Euler-Lagrangian equation yields the optimality condition, where for all $i$,
\begin{equation}\label{eq:optimality_derive}
\E_p[g_i \frac{\partial}{\partial x^i}\log q +     \frac{\partial g_i}{\partial x^i} - \frac{\partial g_i}{\partial x^i}] =\E_p[g_i \frac{\partial}{\partial x^i}\log q] =0.
\end{equation}
\end{proof}

\paragraph{Remarks} It is interesting to know that for Sf-KSD case, the calculus of variation for $f$ does not itself depends on $f$. That is also why we call it optimality condition rather than solving an optimisation problem. Moreover, as we can see from Eq.~\eqref{eq:general_stein_operator}, the role of $\f$ and $\g$ are exchangeable, meaning that if we perturb $g_i$ instead, we will have $f_i$ to satisfy condition $\E_p[f_i \frac{\partial}{\partial x^i}\log q] =0$ for $g_i$ to be a stationary point. 
Hence, the role of $\f$ and $\g$ need to be pre-specified, where we use $\f \in \H^d$ to be our RKHS function and $\g$ to be the auxiliary function. 

We also note that the solution to Eq.\eqref{eq:optimality_derive} is not unique.
As such, the optimality condition is only sufficient condition for stationary point instead of seeking for a global optimal functional.

To further characterise $\g$ from optimality condition $\E_p[g_i \frac{\partial}{\partial x^i}\log q] =0$, we may impose additional constraint for function $\g$. Connecting back to the infinitesimal generator,  $g$ can be viewed as the covariance of diffusion function to be constrained for unit norm diffusion regularisation $\E_p[g_i(x)] = 1$ and  satisfy the optimality condition in Eq.~\eqref{eq:optimality}.

\subsection{Additional Interpretation via Density Ratio as Auxiliary Function}
Density ratio function 
$g = \frac{q}{p}$ naturally satisfy $\E_p[g(x)] = 1$ and $\E_p[g (\log q)'] =0$. However, such ``optimal'' function depends on the alternative distribution $p$ which is unknown in closed form. Although not unique, it can act as a good guideline for choosing useful $g$ function in practice. 

\paragraph{Truncated Gaussian} In the Gaussian distribution in truncated sphere in Eq.~\eqref{eq:truncated_gauss} (results shown in Fig.\ref{tab:power_summary}), the density ratio between the null and the alternative $\frac{q(x)}{p(x)} = \exp\{ \nu x^1 x^2\}$. Projecting to each dimension with $\exp \{\nu x^i\}, i=1,2$, the $L_2$ distance with $g_i^{(p)}(x) = 1 - \|x\|^p$ decreases monotonically when $p$ increases for the cases presented, i.e. polynomial increase is slower than exponential increase. Hence, $g^{(p)}$ with larger $p$ better distinguishes the alternative from the null, yielding higher test power as shown.

\paragraph{Dirichlet Distribution} For Dirichlet distribution example in  
Eq.~\eqref{eq:dirichlet} (results shown in Fig.\ref{tab:power_summary}), the density ration $\frac{q(x)}{p(x)} = (x^{(1)})^{\frac{1}{3}}$. The auxiliary functions used 
are the geometric mean on simplex $g^{(1)}(x) = (x^{(1)} x^{(2)} (1-x^{(1)} -x^{(2)} ))^\frac{1}{3}$ where the first coordinate exactly matches the density ratio, which is closer compared to the Euclidean distance to the nearest boundary which decays quadraticall from the center of the simplex. 

\subsection{Connection to Unnormalised Model Learning Objectives via Density Ratio}\label{app:sm_obj}

Beyond providing a unifying view for KSDs in various testing scenarios, Sf-KSD can also be related to learning objectives for unnormalised models such as score matching \citep{hyvarinen2005estimation},
i.e. score matching can be related to Sf-KSD form with specific choice of kernels and auxiliary functions. 
Recall the score matching objective \citep{hyvarinen2005estimation},
\begin{equation}\label{eq:sm_objective}
    J(p\|q) = \E_p \left[\left(\log p(x)' - \log q(x)'\right)^2\right].
\end{equation}
$J(p\|q)\geq 0$ and the equality holds if and only if $p=q$ under mild regularity conditions \citep{hyvarinen2005estimation}. 

The discrepancy measure in score matching objective \citep{hyvarinen2005estimation} is constructed via the squared difference between derivative of log densities, without the presence of test function, opposing to the KSD-type of discrepancy measure. Hence, instead of just specifying the auxiliary function $g$ to recover various KSDs in the previous results, we also need specific kernel function here to make connections the score matching objective.

\subsubsection*{Sf-KSD Asymptotically Connects to Score Matching}
It was shown that score matching objective cna be viewed as the limit of the diffusion based KSD \citep[Theorem 10]{barp2019minimum}. Here, we derive the result, making explicit connections to density ratio form in Section \ref{sec:optimal_condition}.
\begin{Theorem}
\label{thm:sm_equivalent}
Let $f$, $g$ be scalar functions. Choosing $g(x) = {1}/{\sqrt{p(x)}}$, and exponential quadratic kernel with bandwidth $\sigma$ 
$k_{\sigma}(x,x') = \exp\{ -\frac{1}{\sigma}{\|x-\tilde x\|^2}\}$
, Sf-KSD in the form of Eq.~\eqref{eq:general_stein_operator} converges the score matching objective as $\sigma \to 0$.
\end{Theorem}


\begin{proof}
We rewrite $g$ in the form involving density ratio between $q$ and $p$, such that $$g = \frac{q}{p} \cdot \underset{\xi}{\underbrace{\frac{\sqrt{p}}{q}}}.$$

Then we can write the Sf-KSD of the following form,
\begin{align}
    \operatorname{Sf-KSD}_g (q\|p) 
    &= \E_p[((\log q(x))' (f\xi)(x) + (f\xi)'(x))\frac{q}{p}(x) ] + \E_p[(f\xi)(x) (\frac{q}{p})'(x)] \\
    &= \E_q[((\log q(x))' (f\xi)(x) + (f\xi)'(x))\frac{q}{p}(x)] + \E_p[(f\xi)(x) (\frac{q}{p})'(x)] \\
    &= \E_p[ \frac{f(x)}{\sqrt{p(x)}} \left((\log q(x))' -(\log p(x))'\right)].
\end{align}
The first expectation is $0$ under $q$ from Stein's identity of operator in Eq.~\eqref{eq:general_stein_operator} and the second expectation follows from 
\begin{align*}
 & \left(\frac{q(x)}{p(x)}\right)' = \frac{q'}{p}(x) -\frac{qp'}{p^2}(x) = \frac{q}{p} ((\log q(x))' - (\log p(x))')
\end{align*}

For derivation, we first consider the $\delta$-type function, that we call
$k(x,\tilde x) = { \delta_{x=\tilde x}}
$, we recover the original score-matching objective in \cite{hyvarinen2005estimation} since
\begin{align*}
&\sup_{f \in \mathcal{H}} \E_p[\frac{f(x)}{\sqrt{p(x)}}((\log p(x))' - (\log q(x))')]
=\left\|\E_p \left[k(x,\cdot)\frac{(\log p(x))' - (\log q(x))'}{\sqrt{p(x)}} \right]\right\|^2 \\
&= \E_{x,\tilde x \sim p} \left[\frac{k(x,\tilde x)}{\sqrt{p(x)p(\tilde x)}}((\log p(x))' - (\log q(x))')((\log p(\tilde x))' - (\log q(\tilde x))') \right] \\
&= \int \int \frac{\delta_{x=\tilde{x}} }{\sqrt{p(x)p(\tilde x)}}((\log p(x))' - (\log q(x))')((\log p(\tilde x))' - (\log q(\tilde x))')   p(x) p(\tilde x) dx d\tilde{x} \\
& = \int \frac{1}{p(x)}((\log p(x))' - (\log q(x))')^2 p(x)^2 dx = \E_p[((\log p(x))' - (\log q(x))')^2].
\end{align*}
Using $k_{\sigma}(x, \tilde x) 
\to \delta_{x = \tilde x}$, 
as $\sigma \to 0$, the result follows.
\end{proof}

{We note that the delta function is not bounded even tough integrally bounded by probability density. 
As such, the result appears in the form of limit with vanishing bandwidth.}
As we mentioned in the density ratio argument for optimality conditions, the choice of auxiliary function $g$ here depends on the data density $p$, which is unknown in practice. Score matching relied on the following result for estimating the empirical version of the objective \citep{hyvarinen2005estimation}, 
\begin{equation}\label{eq:sm_square_form}
    J(p\|q) = \E_p \left[\left(\log p(x)' - \log q(x)'\right)^2\right]=
    \E_p \left[ \log q(x)' + \frac{1}{2} \log q(x)''\right].
\end{equation}
\end{document}